\documentclass[a4paper,UKenglish]{lipics}
\usepackage[dvipsnames,svgnames,x11names]{xcolor}
\usepackage{microtype}
\usepackage{pdfpages}
\usepackage{amsmath,amssymb,amsthm,amsopn}
\usepackage{graphicx}
\usepackage{url}
\usepackage{todonotes,soul}
\usepackage{color}
\usepackage{tikz}

\usetikzlibrary{arrows,shapes,snakes,automata,backgrounds,petri}
\usetikzlibrary{positioning,calc}
\usetikzlibrary{patterns}
\tikzstyle{legendborder}=[rectangle, draw, black, rounded corners, thin, top color=white, text=black, minimum width=2.5cm, text width=4.5cm]
\tikzstyle{legendnoborder}=[rectangle, draw, white, rounded corners, thin, top color=white, text=black, minimum width=2.5cm]
\tikzstyle{selected edge} = [draw,line width=1pt,black]
\usetikzlibrary{shapes,backgrounds,plothandlers,plotmarks,calc,arrows,fadings,decorations.pathreplacing,decorations.pathmorphing,}

\tikzstyle{v}=[circle,fill=black,draw=black!75,inner sep=0pt,minimum size=0.3em]
\tikzstyle{I}=[circle,draw=black!75,inner sep=0pt,minimum size=0.8em]
\tikzstyle{J}=[rectangle,draw=black!75,inner sep=0pt,minimum size=0.7em]

\theoremstyle{stylename}
\newtheorem{obs}{Observation}

\newtheorem{assume}{Assumption}
\newtheorem{proposition}{Proposition}
\newtheorem{claim}{Claim}

\usepackage{algorithm}
\usepackage{algorithmic}

\DeclareMathOperator{\reduce}{\sf reduce}
\DeclareMathOperator{\img}{\sf img}

\newcommand{\OO}{{\mathcal O}}

\newtheorem{innerrrx}{Reduction Rule}
\newenvironment{rrx}[1]
  {\renewcommand\theinnerrrx{#1}\innerrrx}
  {\endinnerrrx}
  
\urldef{\mailA}\path|{daniel.lokshtanov,a.mouawad,meirav.zehavi}@ii.uib.no|
\urldef{\mailB}\path|saket@imsc.res.in|
\urldef{\mailC}\path|zehavimeirav@gmail.com|

\title{Packing Cycles Faster Than Erdős-Pósa}

\author[1]{Daniel Lokshtanov}
\author[1]{Amer E. Mouawad}
\author[1,2]{Saket Saurabh}
\author[1]{Meirav Zehavi}
\affil[1]{University of Bergen, Norway\\
\mailA
}
\affil[2]{The Institute of Mathematical Sciences, Chennai, India\\
\mailB
}

\authorrunning{D. Lokshtanov, A. E. Mouawad, S. Saurabh and M. Zehavi}

\Copyright{Daniel Lokshtanov, Amer E. Mouawad, Saket Saurabh and Meirav Zehavi}

\subjclass{G.2.2 Graph Algorithms, I.1.2 Analysis of Algorithms}
\keywords{Parameterized Complexity, Graph Algorithms, Cycle Packing, Erdős-Pósa Theorem}

\serieslogo{}
\volumeinfo
  {Brigitte Vallée and Heribert Vollmer}
  {2}
  {34th International Symposium on Theoretical Aspects of Computer Science (STACS 2017)}
  {1}
  {1}
  {1}
\EventShortName{}

\begin{document}

\maketitle

\begin{abstract}
The {\sc Cycle Packing} problem asks whether a given undirected graph $G=(V,E)$ contains $k$ vertex-disjoint cycles. 
Since the publication of the classic Erdős-Pósa theorem in 1965, this problem received significant scientific attention in the fields of Graph Theory and Algorithm Design. 
In particular, this problem is one of the first problems studied in the framework of Parameterized Complexity. 
The non-uniform fixed-parameter tractability of {\sc Cycle Packing} follows from the Robertson–Seymour theorem, a fact already observed by Fellows and Langston in the 1980s. 
In 1994, Bodlaender showed that {\sc Cycle Packing} can be solved in time $2^{\OO(k^2)}\cdot |V|$ using exponential space. 
In case a solution exists, Bodlaender's algorithm also outputs a solution (in the same time). It has later become 
common knowledge that {\sc Cycle Packing} admits a $2^{\OO(k\log^2k)}\cdot |V|$-time (deterministic) algorithm using exponential space, which is a consequence 
of the Erdős-Pósa theorem. Nowadays, the design of this algorithm is given as an exercise in textbooks on Parameterized Complexity. Yet, no algorithm that 
runs in time $2^{o(k\log^2k)}\cdot |V|^{\OO(1)}$, beating the bound $2^{\OO(k\log^2k)}\cdot |V|^{\OO(1)}$, has been found. 
In light of this, it seems natural to ask whether the $2^{\OO(k\log^2k)}\cdot |V|^{\OO(1)}$ bound is essentially optimal. 
In this paper, we answer this question negatively by developing a $2^{\OO(\frac{k\log^2k}{\log\log k})}\cdot |V|$-time (deterministic) algorithm for {\sc Cycle Packing}. 
In case a solution exists, our algorithm also outputs a solution (in the same time).  Moreover, apart from beating 
the bound $2^{\OO(k\log^2k)}\cdot |V|^{\OO(1)}$, our algorithm runs in time linear in $|V|$, and its space complexity is polynomial in the input size. 
\end{abstract}

%
%

\section{Introduction}\label{sec:intro}

The {\sc Cycle Packing} problem asks whether a given undirected graph $G=(V,E)$ contains $k$ vertex-disjoint cycles. Since the publication of the classic Erdős-Pósa theorem in 1965~\cite{ErdosPosa}, this problem received significant scientific attention in the fields of Graph Theory and Algorithm Design. In particular, {\sc Cycle Packing} is one of the first problems studied in the framework of Parameterized Complexity.
In this framework, each problem instance is associated with a parameter $k$ that is a non-negative integer, and a problem is said to be {\em fixed-parameter tractable (FPT)} if the combinatorial explosion in the time complexity can be confined to the parameter $k$. More precisely, a problem is FPT if it can be solved in time $f(k)\cdot |I|^{\OO(1)}$ for some function $f$, where $|I|$ is the input size. For more information, we refer the reader to recent monographs such as \cite{ParamBook13} and \cite{ParamBook15}. 

In this paper, we study the {\sc Cycle Packing} problem from the perspective of Parameterized Complexity. In the standard parameterization of {\sc Cycle Packing}, the parameter is the number $k$ of vertex-disjoint cycles. The non-uniform fixed-parameter tractability of {\sc Cycle Packing} follows from the Robertson–Seymour theorem \cite{RobSeyDisPat},\footnote{The paper \cite{RobSeyDisPat} was already available as a manuscript in 1986 (see, e.g., \cite{OnDisCycBod}).} a fact already observed by Fellows and Langston in the 1980s. In 1994, Bodlaender showed that {\sc Cycle Packing} can be solved in time $2^{\OO(k^2)}\cdot |V|$ using exponential space \cite{OnDisCycBod}. Notably, in case a solution exists, Bodlaender's algorithm also outputs a solution in time $2^{\OO(k^2)}\cdot |V|$.

The Erdős-Pósa theorem states that there exists a function $f(k)=\OO(k\log k)$ such that for each non-negative integer $k$, every undirected graph either contains $k$ vertex-disjoint cycles or it has a feedback vertex set consisting of $f(k)$ vertices \cite{ErdosPosa}. It is well known that the treewidth ($\mathrm{\it tw}$) of a graph is not larger than its feedback vertex set number ($\mathrm{\it fvs}$), and that a naive dynamic programming (DP) scheme solves {\sc Cycle Packing} in time $2^{\OO(\mathrm{\it tw}\log \mathrm{\it tw})}\cdot|V|$ and exponential space (see, e.g., \cite{ParamBook15}). Thus, the existence of a $2^{\OO(k\log^2k)}\cdot|V|$-time (deterministic) algorithm that uses exponential space can be viewed as a direct consequence of the Erdős-Pósa theorem. Nowadays, the design of this algorithm is given as an exercise in textbooks on Parameterized Complexity such as \cite{ParamBook13} and \cite{ParamBook15}. In case a solution exists, this algorithm does not output a solution (though we remark that with a certain amount of somewhat non-trivial work, it is possible to modify this algorithm to also output a solution).

Prior to our work, no algorithm that runs in time $2^{o(k\log^2k)}\cdot|V|^{\OO(1)}$, beating the bound $2^{\OO(k\log^2k)}\cdot|V|^{\OO(1)}$, has been found. 
In light of this, it seemed tempting to ask whether the $2^{\OO(k\log^2k)}\cdot|V|^{\OO(1)}$ bound is essentially optimal. 
In particular, two natural directions to explore in order to obtain a faster algorithm necessarily lead to a dead end. 
First, Erdős and Pósa \cite{ErdosPosa} proved that the bound $f(k)=\OO(k\log k)$ in their theorem is essentially tight as there 
exist infinitely many graphs and a constant $c$ such that these graphs do not contain $k$ vertex-disjoint cycles 
and yet their feedback vertex set number is larger than $ck\log k$. Second, Cygan et al.~\cite{Cut&Count} proved 
that the bound $2^{\OO(\mathrm{\it tw}\log \mathrm{\it tw})}\cdot|V|^{\OO(1)}$ is also likely to be essentially tight 
in the sense that unless the Exponential Time Hypothesis (ETH) \cite{ETH} is false, {\sc Cycle Packing} cannot be solved 
in time $2^{o(\mathrm{\it tw}\log \mathrm{\it tw})}\cdot|V|^{\OO(1)}$ (however, it might still be true that {\sc Cycle Packing} is 
solvable in time $2^{o(\mathrm{\it fvs}\log \mathrm{\it fvs})}\cdot|V|^{\OO(1)}$).

\subsection{Related Work}

The {\sc Cycle Packing} problem admits a factor $\OO(\log|V|)$ approximation algorithm \cite{CycPacApprox}, and it is quasi-NP-hard to approximate within a factor of $\OO(\log^{\frac{1}{2}-\epsilon}|V|)$ for any $\epsilon>0$ \cite{CycPacNoApprox}. In the context of kernelization with respect to the parameter $k$, {\sc Cycle Packing} does not admit a polynomial kernel unless NP $\subseteq$ coNP/Poly \cite{CycPacNoKernel}. Recently, Lokshtanov et al.~\cite{CycPacApproxKernel} obtained a 6-approximate kernel with $\OO((k\log k)^2)$ vertices along with a $(1 + \epsilon)$-approximate kernel with $k^{O(f(\epsilon)\log k)}$ vertices for some function $f$.  We would like to mention that in case one seeks $k$ edge-disjoint cycles rather than $k$ vertex-disjoint cycles, the problem becomes significantly simpler in the sense that it admits a kernel with $\OO(k\log k)$ vertices \cite{CycPacNoKernel}.

Focusing on structural parameters, Bodlaender et al.~\cite{CycPacOtherParams} obtained polynomial kernels with respect to the vertex cover number, vertex-deletion distance to a cluster graph and the max leaf number. In planar graphs, Bodlaender et al.~\cite{BodIS08} solved {\sc Cycle Packing} in subexponential time $2^{\OO(\sqrt{k})}\cdot|V|^{\OO(1)}$, and showed that this problem admits a linear kernel. In the more general class of $H$-minor-free graphs, Dorn et at.~\cite{CycPacMinorFree} also solved {\sc Cycle Packing} in subexponential time $2^{\OO(\sqrt{k})}\cdot|V|^{\OO(1)}$. Moreover, for apex-minor-free graphs, Fomin et al.~\cite{CycPacLinKernel} showed that {\sc Cycle Packing} admits a linear kernel, and Fomin et al.~\cite{CycPacEPTAS} showed that it also admits an EPTAS. When the input graph is a directed graph, {\sc Cycle Packing} is W[1]-hard \cite{CycPacW1Hard}, but it admits an FPT approximation scheme \cite{CycPacParamApprox}. In fact, {\sc Cycle Packing} in directed graphs was the first W[1]-hard problem shown to admit such a scheme. Krivelevich et al.~\cite{CycPacApprox} obtained a factor $\OO(|V|^{\frac{1}{2}})$ approximation algorithm for {\sc Cycle Packing} in directed graphs and showed that this problem is quasi-NP-hard to approximate within a factor of $\OO(\log^{1-\epsilon}|V|)$ for any $\epsilon>0$.

Several variants of {\sc Cycle Packing} have also received significant scientific attention. 
For example, the variant of {\sc Cycle Packing} where one seeks $k$ {\em odd} vertex-disjoint cycles has been 
widely studied \cite{Reed99,Tho01,RautReed01,KawWollan06,KawNa07,KawReed09}. Another well-known variant, where the cycles need 
to contain a prescribed set of vertices, has also been extensively investigated \cite{KakKawMarx11,PonWollan12,KawKob12,KakKaw13,KawKra13}. 
Furthermore, a combination of these two variants has been considered in~\cite{KakKaw13,Joos16}.

Finally, we briefly mention that inspired by the Erdős-Pósa theorem, a class of graphs $\cal H$ is said to have the Erdős-Pósa property if there is a function $f(k)$ for which given a graph $G$, it either contains $k$ vertex-disjoint subgraphs such that each of these subgraphs is isomorphic to a graph in $\cal H$, or it contains a set of $f(k)$ vertices that hits each of its subgraphs that is isomorphic to a graph in $\cal H$. A fundamental result in Graph Theory by Robertson and Seymour \cite{RobSey} states the the class of all graphs that can be contracted to a fixed planar graph $H$ has the Erdős-Pósa property. Recently, Chekuri and Chuzhoy \cite{CC-STOC13} presented a framework that leads to substantially improved functions $f(k)$ in the context of results in the spirit of the Erdős-Pósa theorem.
Among other results, these two works are also related to the recent breakthrough result by Chekuri and Chuzhoy \cite{Grid1}, which states that every graph of treewidth at least $f(k)=\OO(k^{98}\cdot\mathrm{polylog}(k))$ contains the $k\times k$ grid as a minor (the constant $98$ has been improved to $36$ in \cite{Grid2} and to $19$ in \cite{Grid3}).
Following the seminal work by Robertson and Seymour \cite{RobSey}, numerous papers (whose survey is beyond the scope of this paper) investigated which other classes of graphs have the Erdős-Pósa property, which are the ``correct'' functions $f$ associated with them, and which generalizations of this property lead to interesting discoveries.

\subsection{Our Contribution}

In this paper, we show that the running time of the algorithm  that is a consequence of the Erdős-Pósa theorem is not essentially tight. For this purpose, we develop a $2^{\OO(\frac{k\log^2k}{\log\log k})}\cdot|V|$-time (deterministic) algorithm for {\sc Cycle Packing}. In case a solution exists, our algorithm also outputs a solution (in time $2^{\OO(\frac{k\log^2k}{\log\log k})}\cdot |V|$). Moreover, apart from beating the bound $2^{\OO(k\log^2k)}\cdot |V|^{\OO(1)}$, our algorithm runs in time linear in $|V|$, and its space complexity is polynomial in the input size. Thus, we also improve upon the classical $2^{\OO(k^2)}\cdot |V|$-time algorithm by Bodlaender \cite{OnDisCycBod}. Our result is summarized in the following theorem.

\begin{theorem}\label{thm:main}
There exists a (deterministic) polynomial-space algorithm that solves {\sc Cycle Packing} in time $2^{\OO(\frac{k\log^2k}{\log\log k})}\cdot|V|$. In case a solution exists, it also outputs a solution.
\end{theorem}

Our technique relies on several combinatorial arguments that might be of independent interest, and whose underlying ideas might be relevant to the design of other parameterized algorithms. Let us now outline the structure of our proof, specifying the main ingredients that we develop along the way.
\begin{itemize}
\item First, we show that in time linear in $|V|$, it is easy to bound $|E|$ by $\OO(k\log k\cdot |V|)$ (Assumption~\ref{assume:numEdges}).
\item Second, we give an algorithmic version of the Erdős-Pósa theorem that runs in time linear in $|V|$ and which outputs either a solution or a small feedback vertex set (Theorem~\ref{thm-erdosposa}).
\item Then, we show that given a graph $G=(V,E)$ and a feedback vertex set $F$, a shortest cycle in $G$ can be found in time $\OO(|F|\cdot(|V|+|E|))$ (Lemma~\ref{thm-itai}).
\item We proceed by interleaving an application of a simple set of reduction rules (Rules \ref{rr-degone}, \ref{rr-degtwo} and \ref{rr-multi}) with a computation of a ``short'' cycle. Thus, given some $g>6$, we obtain a set $S$ of size smaller than $gk$ such that the girth of the ``irreducible component'' of $G-S$ is larger than $g$ (Lemma \ref{lem-packing}). Here, the irreducible component of $G-S$ is the graph obtained from $G-S$ by applying our reduction rules.
\item Next, we show that the number of vertices in the above mentioned irreducible component is actually ``small'' --- for some fixed constant $c$, it can be bounded by $(2ck \log k)^{1 + {6 \over g-6}} + 3ck\log k$ (Lemma \ref{lem-packing-size}). The choice of $g = {48 \log k \over \log \log k} + 6$ results in the bound $3ck\log k + 2ck\log^{1.5} k$ (Corollary \ref{cor-packing-size}).\footnote{We found these constants as the most natural ones to obtain a clean proof of {\em any} bound of the form $\OO(\frac{k\log^2k}{\log\log k})$ (that is, the constants were not optimized to obtain the bound $3ck\log k + 2ck\log^{1.5} k$).}
\item Now, we return to examine the graph $G-S$ rather than only its irreducible component. The necessity of this examination stems from the fact that our reduction rules, when applied to $G-S$ rather than $G$, do not preserve solutions. We first give a procedure which given any set $X$, modifies the graph $G-X$ in a way that both preserves solutions and gets rid of many leaves (Lemma \ref{lem-lossy}). We then use this procedure to bound the number of leaves, as well as other ``objects'', in the reducible component of $G-S$ (Lemma \ref{lem-core}).
\item At this point, the graph $G$ may still contain many vertices: the reducible component of $G-S$ may contain ``long'' induced paths (which are not induced paths in  $G$). We show that the length of these paths can be shortened by ``guessing'' permutations that provide enough information describing the relations between these paths and the vertices in $S$. Overall, we are thus able to bound the entire vertex-set of $G$ by $\OO(k\log^{1.5}k)$ in time $2^{\OO(\frac{k\log^2k}{\log\log k})}\cdot|V|$ and polynomial space (Lemma~\ref{thm:permute}).
\item Finally, we apply a DP scheme (Lemma \ref{lem:DP}). Here, to ensure that the space complexity is polynomial in the input size, we rely on the principle of inclusion-exclusion.
\end{itemize}


\section{Preliminaries}\label{sec:prelim}

We use standard terminology from the book of Diestel~\cite{diestel-book} for those graph-related terms that are not explicitly defined here. We only consider finite graphs possibly having self-loops and multi-edges. Moreover, we restrict the maximum multiplicity of an edge to be 2. For a graph $G$, we use $V$ and $E$ to denote the vertex and edge sets of the graph $G$, respectively. For a vertex $v \in V$, we use $\text{deg}_G(v)$ to denote the degree of $v$, i.e the number of edges incident on $v$, in the (multi) graph $G$. We also use the convention that a self-loop at a vertex $v$ contributes 2 to its degree. For a vertex subset $S \subseteq V$, we let $G[S]$ and $G - S$ denote the graphs induced on $S$ and $V \setminus S$, respectively. For a vertex subset $S \subseteq V$, we use $N_G(S)$ and $N_G[S]$ to denote the open and closed neighborhoods of $S$ in $G$, respectively. That is, $N_G(S) = \{v \mid \{u,v\} \in E, u \in S\} \setminus S$ and $N_G[S] = N_G(S) \cup S$. In case $S=\{v\}$, we simply let $N(v)=N(S)$ and $N[v]=N[S]$. For a graph $G = (V, E)$ and an edge $e \in E$, we let $G / e$ denote the graph obtained by contracting $e$ in $G$.  For $E' \subseteq {V \choose 2}$, i.e. a subset of edges, we let $G + E'$ denote the (multi) graph  obtained after adding the edges in $E'$ to $G$, and we let $G / E'$ denote the (multi) graph  obtained after contracting the edges of $E'$ in $G$. The girth of a graph is denoted by $\text{girth}(G)$, its minimum degree by $\delta(G)$, and its maximum degree by $\Delta(G)$. A graph with no cycles has infinite girth.

A \emph{path} in a graph is a sequence of distinct vertices $v_0, v_1, \ldots, v_\ell$ such that $\{v_i,v_{i+1}\}$ is an edge for all 
$0 \leq i < \ell$. A \emph{cycle} in a graph is a sequence of distinct vertices $v_0, v_1, \ldots, v_\ell$ such that 
$\{v_i,v_{(i+1) \mod \ell + 1}\}$ is an edge for all $0 \leq i \leq \ell$. 
Both a double edge and a self-loop are cycles. If $P$ is a path from a vertex $u$ to a vertex $v$ in the graph 
$G$ then we say that $u$ and $v$ are the end vertices of the path $P$ and $P$ is a $(u,v)$-path. For a path $P$, we use $V(P)$ and $E(P)$ to 
denote the sets of vertices and edges in the path $P$, respectively, and length of $P$ is denoted by $|P|$ (i.e, $|P| = |V(P)|$). 
For a cycle $C$, we use $V(C)$ and $E(C)$ to denote the sets of vertices and edges in the cycle $C$, respectively, and the length of 
$C$, denoted by $|C|$, is $|V(C)|$. For a path or a cycle $Q$ we use $N_G(Q)$ and $N_G[Q]$ to denote the sets $N_G(V(Q))$ and $N_G[V(Q)]$, respectively. 
For a collection of paths/cycles $\mathcal{Q}$, we use $|\mathcal{Q}|$ to denote the number of paths/cycles in $\mathcal{Q}$ and $V(\mathcal{Q})$ to denote 
the set $\bigcup_{Q \in \mathcal{Q}} V(Q)$. We say a path $P$ in $G$ is a {\em degree-two path} if all vertices in $V(P)$, including the 
end vertices of $P$, have degree exactly 2 in $G$. We say $P$ is a {\em maximal degree-two path} if no proper superset of $P$ also forms a degree-two path. 
We note that the notions of {\em walks} and {\em closed walks} are defined exactly as paths and cycles, respectively, except that their vertices need 
not be distinct. Finally, a {\em feedback vertex set} is a subset $F$ of vertices such that $G - F$ is a forest.

Below we formally state some of the key results that will be used throughout the paper, starting with the classic Erdős-Pósa theorem \cite{ErdosPosa}.

\begin{proposition}[\cite{ErdosPosa}]\label{thm-erdosposa1}
There exists a constant $c'$ such that every (multi) graph either contains $k$ vertex-disjoint cycles or it has a feedback vertex set of size at most $c'k \log k$.
\end{proposition}

Observe that any (multi) graph $G=(V,E)$ whose feedback vertex set number is bounded by $c'k \log k$ has less than $(2c'k\log k+1)\cdot |V|$ edges 
(recall that we restrict the multiplicity of an edge to be 2). Indeed, letting $F$ denote a feedback vertex set of minimum size, the worst 
case (in terms of $|E|$) is obtained when $G-F$ is a tree, which contains $|V|-|F|-1$ edges, and between every pair 
of vertices $v\in F$ and $u\in V$, there exists an edge of multiplicity 2. 
Thus, by Proposition \ref{thm-erdosposa1}, in case $|E|>(2c'k\log k+1)\cdot |V|$, the input instance is a 
yes-instance, and after we discard an arbitrary set of $|E|-(2c'k\log k+1)\cdot |V|$ edges, it remains a yes-instance. 
A simple operation which discards at least $|E|-(2c'k\log k + 1)\cdot |V|$ edges 
and can be performed in time $\OO(k\log k\cdot|V|)$ is described in Appendix \ref{app:discard}. 

\begin{assume}\label{assume:numEdges}
We assume that $|E|=\OO(k\log k\cdot |V|)$.
\end{assume}

Now, we state our algorithmic version of Proposition \ref{thm-erdosposa1}. The proof partially builds upon the proof of the Erdős-Pósa theorem in the book~\cite{diestel-book}, and it is given in Appendix~\ref{app:ErdosPosa}. 

\begin{theorem}\label{thm-erdosposa}
There exists a constant $c$ and a polynomial-space algorithm such that given a (multi) graph $G$ and a non-negative integer $k$, in time $k^{\OO(1)}\cdot|V|$ it either outputs $k$ vertex-disjoint cycles or a feedback vertex set of size at most $ck \log k=r$.
\end{theorem}

Next, we state two results relating to cycles of average and short lengths.

\begin{proposition}[\cite{Alon}]\label{thm-alonetal}
Any (multi) graph $G = (V, E)$ on $n$ vertices with average degree $d$ contains a cycle of length at most $2 \log_{d-1} n + 2$.
\end{proposition}

Itai and Rodeh \cite{Itai} showed that given a (multi) graph $G = (V, E)$, an ``almost'' shortest cycle (if there is any) in $G$ can be found in time $\OO(|V|^2)$. To obtain a linear dependency on $|V|$ (given a small feedback vertex set), we prove the following result in Appendix \ref{app:thm-itai}.

\begin{lemma}\label{thm-itai}
Given a (multi) graph $G = (V, E)$ and a feedback vertex set $F$ of $G$, a shortest cycle (if there is any) in $G$ can be found in time $\OO(|F|\cdot(|V|+|E|))$.
\end{lemma}

Finally, we state a result that will be used (in Lemma \ref{lem-packing-size}) to bound the size of a graph we obtain after performing simple preprocessing operations as well as repetitive removal of short cycles.

\begin{proposition}[\cite{FVS}, Lemma 9]\label{thm-saket}
Let $T = (V, E)$ be a forest on $N$ vertices. Let $M' = \{\{i, j\} \in E \mid \text{deg}_T(i) = \text{deg}_T(j) = 2\}$ and $L = \{a \in V \mid \text{deg}_T(a) \leq 1\}$. Then there exists $M \subseteq M'$ such that $M$ is a matching and $|W| \geq {N \over4}$ where $W=L\cup M$.
\end{proposition}


\section{Removing Leaves, Induced Paths, and Short Cycles}\label{sec:remove}

As is usually the case when dealing with cycles in a graph, we first define three rules which help getting rid of vertices of degree at most 2 as well as edges of multiplicity larger than 2. It is not hard to see that all three Reduction Rules~\ref{rr-degone},~\ref{rr-degtwo}, and~\ref{rr-multi} are safe, i.e. they preserve solutions in the reduced graph.

\begin{rrx}{A1}\label{rr-degone}
Delete vertices of degree at most 1.
\end{rrx}

\begin{rrx}{A2}\label{rr-degtwo}
If there is a vertex $v$ of degree exactly 2 that is not incident to a self-loop, then delete $v$ and connect its two (not necessarily distinct) neighbors by a new edge.
\end{rrx}

\begin{rrx}{A3}\label{rr-multi}
If there is a pair of vertices $u$ and $v$ in $V$ such that $\{u,v\}$ is an edge of multiplicity larger than 2, then reduce the multiplicity of the edge to 2.
\end{rrx}

Observe that the entire process that applies these rules exhaustively can be done in time $\OO(|V|+|E|)=\OO(k\log k\cdot|V|)$. 
Indeed, in time $\OO(|V|)$ we first remove the vertex-set of each maximal path between a leaf and a degree-two vertex. 
No subsequent application of Rule \ref{rr-degtwo} or Rule \ref{rr-multi} creates vertices of degree at most one. 
Now, we iterate over the set of degree-two vertices. For each degree-two vertex that is not incident to a self-loop, we apply Rule \ref{rr-degtwo}. 
Next, we iterate over $E$, and for each edge of multiplicity larger than two, we apply Rule \ref{rr-multi}. 
At this point, the only new degree-two vertices that can be created are vertices incident to exactly one edge, whose multiplicity is two. 
Therefore, during one additional phase where we exhaustively apply Rule \ref{rr-degtwo}, the only edges of multiplicity larger than two that can be created are self-loops. 
Thus, after one additional iteration over $E$, we can ensure that no rule among Rules \ref{rr-degone}, \ref{rr-degtwo} and \ref{rr-multi} is applicable.

Since these rules will be applied dynamically and iteratively, we define an operator, denoted by $\reduce(G)$, that takes as input a graph $G$ and returns the (new) graph $G'$ that results from an exhaustive application of Rules~\ref{rr-degone},~\ref{rr-degtwo} and~\ref{rr-multi}.

\begin{definition}
For a (multi) graph $G$, we let $G' = \reduce(G)$ denote the graph obtained after an exhaustive application of Reduction Rules~\ref{rr-degone},~\ref{rr-degtwo} and~\ref{rr-multi}. $|\reduce(G)|$ denotes the number of vertices in $\reduce(G)$. Moreover, $\img(\reduce(G))$ denotes the pre-image of $\reduce(G)$, i.e. $\img(\reduce(G))$ is the set of vertices in $G$ which are not deleted in $\reduce(G)$.
\end{definition}

\begin{obs}\label{obs-edges}
For a graph $G = (V, E)$ and a set $E' \subseteq {V \choose 2}$ it holds that $|\reduce(G + E')| \leq |\reduce(G)| + 2|E'|$.
\end{obs}

The first step of our algorithm consists of finding, in time linear in $|V|$, a set $S$ satisfying the conditions specified in Lemmata~\ref{lem-packing} and~\ref{lem-packing-size}.
Intuitively, $S$ will contain vertices of ``short'' cycles in the input graph, where short will be defined later.

\begin{lemma}\label{lem-packing}
Given a (multi) graph $G = (V, E)$ and two integers $k > 0$ and $g > 6$, there exists an $k^{\OO(1)}\cdot|V|$-time algorithm that either finds $k$ vertex-disjoint cycles in $G$ or finds a (possibly empty) set $S \subseteq V$ such that $\text{girth}(\reduce(G - S)) > g$ and $|S| < gk$.
\end{lemma}

\begin{proof}
We proceed by constructing such an algorithm. First, we apply the algorithm of Theorem~\ref{thm-erdosposa} which outputs either $k$ vertex-disjoint 
cycles or a feedback vertex set $F$ of size at most $ck \log k = r$. In the former case we are done. In the latter case, i.e. the case where 
a feedback vertex set $F$ is obtained, we apply the following procedure iteratively (initially, we set $S = \emptyset$):

\begin{itemize}
\item[] (1) Apply Lemma~\ref{thm-itai} to find a shortest cycle $C$ in $\reduce(G)$.
\item[] (2) If no cycle was found or $|C| > g$ then return $S$.
\item[] (3) Otherwise, i.e. if $|C| \leq g$, then add the vertices of $C$ to $S$, delete those vertices from $G$ to obtain $G'$, set $G = G'$, and repeat from Step (1).
\end{itemize}

\noindent
Note that if Step (3) is applied $k$ times then we can terminate and return the corresponding $k$ vertex-disjoint cycles in $G$. Hence, when the condition of Step (2) is satisfied, i.e. when the described procedure terminates, the size of $S$ is at most $g(k - 1) < gk$ and $\text{girth}(\reduce(G - S)) > g$. Since the algorithm of Theorem~\ref{thm-erdosposa} runs in time $k^{\OO(1)}\cdot|V|$, and each iteration of Steps (1)-(3) is performed in time $\OO((k\log k)^2\cdot|V|)$, we obtain the desired time complexity.
\end{proof}

\begin{lemma}\label{lem-packing-size}
Given a (multi) graph $G = (V, E)$ and two integers $k > 0$ and $g > 6$, let $S$ denote the set obtained after applying the algorithm of Lemma~\ref{lem-packing} (assuming no $k$ vertex-disjoint cycles obtained). Then $|\reduce(G - S)| \leq (2ck \log k)^{1 + {6 \over g-6}} + 3ck\log k$.
\end{lemma}

\begin{proof}
Let $G' = (V', E') = \reduce(G - S)$ and $|V'| = n'$. First, recall that $G$ admits a feedback vertex set of size at most $ck \log k = r$. 
Since Reduction Rules \ref{rr-degone}, \ref{rr-degtwo} and \ref{rr-multi} do not increase the feedback vertex set of the graph (see, e.g., \cite{FVS}, Lemma 1), $G'$ also 
admits a feedback vertex set $F$ of size at most $r$. Let $T$ denote the induced forest on the remaining $N = n' - r$ vertices in $G'$. Moreover, from 
Lemma~\ref{lem-packing}, we know that $\text{girth}(G') > g > 6$.

Next, we apply Proposition~\ref{thm-saket} to $T$ to get $W$. Now with every element $a \in W$ we associate an unordered pair of vertices of $F$ as follows.
Assume $a \in L$, i.e. $a$ is a vertex of degree 0 or 1. Since the degree of $a$ is at least 3 in $G'$, $a$ has at least two neighbors in $F$. We pick two of these neighbors arbitrarily and associate them with $a$. We use $\{x_a, y_a\}$ to denote this pair. If $a = \{u, v\}$ is an edge from $M$ then each of $u$ and $v$ has degree at least 3 in $G'$ and each has at least one neighbor in $F$. We pick one neighbor for each and associate the pair $\{x_u, x_v\}$ with $a$. Note that since $\text{girth}(G') > 6$, $x_u \neq x_v$ and $x_a \neq y_a$.

We now construct a new multigraph $G^\star = (V^\star, E^\star)$ with vertex set $V^\star = F$ as follows. For every vertex $a \in W$ we include an edge in $E^\star$ between $x_a$ and $y_a$, and for every edge $a = \{u, v\} \in W$ we include an edge in $E^\star$ between $x_u$ and $x_v$. By Proposition~\ref{thm-saket}, we know that $W$ is of size at least ${N \over 4}$. It follows that $G^\star$ has at least ${N \over 4}$ edges and hence its average degree is at least ${N \over 2r}$ as $|V^\star| = ck \log k = r$. Note that if $G^\star$ has a cycle of length at most $\ell$, then $G'$ has a cycle of length at most $3\ell$, as any edge of the cycle in $G^\star$ can be replaced by a path of length at most $3$ in $G'$. Combining this with the fact
that $\text{girth}(G') > g > 6$, we conclude that $G^\star$ contains no self-loops or parallel edges. Hence $G^\star$ is a simple
graph with average degree at least ${N \over 2r}$. By Proposition~\ref{thm-alonetal}, $G^\star$ must have a cycle of length at most
$$ 2 \log_{{N \over 2r}-1} r + 2 = {2 \log r \over \log ({N \over 2r}-1)} + 2$$
which implies that $G'$ must have a cycle of length at most
$$ {6 \log r \over \log ({N \over 2r}-1)} + 6.$$
Finally, by using the fact that $\text{girth}(G') > g$ and substituting $N$ and $r$, we get
\[\begin{array}{ll}
\displaystyle{{6 \log r \over \log ({N \over 2r}-1)} + 6 > g} & \displaystyle{\Longleftrightarrow \log r  > {(g-6) \over 6} \log \big({N - 2r \over 2r}\big)}\\
& \displaystyle{\Longleftrightarrow \log r  > {(g-6) \over 6} \log (N - 2r) - {(g-6) \over 6}\log (2r)}\\
& \displaystyle{\Longleftrightarrow {\log r  + {(g-6) \over 6} \log (2r) \over {(g-6) \over 6}} > \log (N - 2r)}\\
& \displaystyle{\Longrightarrow {\log (2r)  + {(g-6) \over 6} \log (2r) \over {(g-6) \over 6}} > \log (N - 2r)}\\
& \displaystyle{\Longleftrightarrow (1+{6 \over g-6})\log (2r) > \log (N - 2r)}\\
& \displaystyle{\Longleftrightarrow (1+{6 \over g-6})\log (2ck\log k)  > \log (n' - 3ck\log k)}\\
& \displaystyle{\Longleftrightarrow (2ck\log k)^{1 + {6 \over (g-6)}} + 3ck\log k > n'}.
\end{array}\]
This completes the proof.
\end{proof}

The usefulness of Lemma~\ref{lem-packing-size} comes from the fact that by setting $g = {48 \log k \over \log \log k} + 6$, we can guarantee that $|\reduce(G - S)| < 3ck\log k + 2ck\log^{1.5} k$, and therefore we can beat the $\OO(k\log^{2} k)$ bound. That is, we have the following consequence.

\begin{corollary}\label{cor-packing-size}
Given a (multi) graph $G = (V, E)$ and an integer $k > 0$, let $S$ denote the set obtained after applying the algorithm of Lemma~\ref{lem-packing} with $g={48 \log k \over \log \log k} + 6$ (assuming no $k$ vertex-disjoint cycles obtained). Then $|\reduce(G - S)| \leq 3ck\log k + 2ck\log^{1.5} k$.
\end{corollary}

\begin{proof}
By Lemma \ref{lem-packing-size}, $|\reduce(G - S)| \leq (2ck \log k)^{1 + {\log\log k \over 8\log k}} + 3ck\log k$.
Assuming $k > \log k > c > 2$, we have $(2ck \log k)^{1 + {\log\log k \over 8\log k}} =
(2ck \log k) (2ck \log k)^{{\log\log k \over 8\log k}} \leq (2ck \log k) k^{{4\log\log k \over 8\log k}}$.
Now note that $k^{{4\log\log k \over 8\log k}} \leq \log^{0.5} k$.
Hence, $(2ck \log k)^{1 + {\log\log k \over 8\log k}} \leq 2ck \log k \log^{{1 \over 2}} k \leq 2ck \log^{1.5} k$.
This completes the proof. 
\end{proof} 


\section{Bounding the Core of the Remaining Graph}\label{sec:bound}

At this point, we assume, without loss of generality, that we are given a graph $G = (V, E)$, a positive 
integer $k$, $g = {48 \log k \over \log \log k} + 6$, and a 
set $S \subseteq V$ such that $\text{girth}(\reduce(G - S)) > g$, $|S| < gk$, and 
$|\reduce(G - S)| \leq 3ck\log k + 2ck\log^{1.5} k$.

Even though the number of vertices in $\reduce(G - S)$ is bounded, the number of vertices in $G - S$ is unbounded. 
In what follows, we show how to bound the number of ``objects'' in $G - S$, where an object is either a 
vertex in $G - S$ or a degree-two path in $G - S$. 
The next lemma is a refinement extending a lemma by Lokshtanov et al.~\cite{CycPacApproxKernel} (Lemma 5.2).
We give a full proof in Appendix~\ref{app:bound}. 

\begin{lemma}\label{lem-lossy}
Let $G = (V, E)$ be a (multi) graph and let $X \subseteq V$ be any subset of the vertices of $G$. 
Suppose there are more than $|X|^2(2|X| + 1)$ vertices in $G - X$ whose degree in $G - X$ is at most one. 
Then, there is either an isolated vertex $w$ in $G - X$ or an edge $e \in E$ such that $(G, k)$ is a yes-instance  of {\sc Cycle Packing} if and 
only if either $(G - \{w\}, k)$ or $(G/e, k)$ is a yes-instance. Moreover,
there is an $\OO(|X|^2\cdot k\log k\cdot |V|)$-time algorithm that given $G$ and $X$, outputs sets $V_X\subseteq V\setminus X$ and $E_X\subseteq E(G-X)$ such 
that, for the graph $G'=(G/E_X) - V_X$, it holds that $(G, k)$ is a yes-instance of {\sc Cycle Packing} if and only if $(G', k)$ is 
a yes-instance of {\sc Cycle Packing}, and $G'-X$ contains at most $|X|^2(2|X| + 1)$ vertices whose degree in $G' - X$ is at most one. 
\end{lemma}


Armed with Lemma~\ref{lem-lossy}, we are now ready to prove the following result. 
For a forest $T$, we let $T_{\leq 1}$, $T_{2}$, and $T_{\geq 3}$, denote the sets of 
vertices in $T$ having degree at most one in $T$, degree exactly two in $T$, and degree 
larger than two in $T$, respectively. Moreover, we let $\mathcal{P}$ denote the set of all maximal degree-two paths in $T$.

\begin{figure}[htp]
    \centering
    \begin{tikzpicture}[scale=.5, auto=left, remember picture ,every node/.style={circle},inner/.style={circle},outer/.style={circle}]


    \node[outer] (S) at (8,8) [ellipse, fill=white, draw=black] {
    \begin{tikzpicture}
        \node [inner, circle, fill=black, draw=black, inner sep=1.5pt] (v1) {};
        \node [inner, circle, fill=black, draw=black, inner sep=1.5pt, right=0.3cm of v1] (v2) {};
        \node [inner, circle, fill=black, draw=black, inner sep=1.5pt, below=0.3cm of v1] (v3) {};
		\node [inner, circle, fill=black, draw=black, inner sep=1.5pt, right=0.3cm of v3] (v4) {};
		                                                        
        \node [inner, circle, fill=black, draw=black, inner sep=1.5pt, right=0.5cm of v2] (v5) {};
        \node [inner, circle, fill=black, draw=black, inner sep=1.5pt, right=0.3cm of v5] (v6) {};
        \node [inner, circle, fill=black, draw=black, inner sep=1.5pt, below=0.3cm of v5] (v7) {};
		\node [inner, circle, fill=black, draw=black, inner sep=1.5pt, right=0.3cm of v7] (v8) {};
		                                                        
        \node [inner, circle, fill=black, draw=black, inner sep=1.5pt, right=0.65cm of v6] (v9) {};
        \node [inner, circle, fill=black, draw=black, inner sep=1.5pt, right=0.5cm of v8] (v10) {};
        \node [inner, circle, fill=black, draw=black, inner sep=1.5pt, right=0.3cm of v10] (v11) {};
		                                                        
		\node [inner, circle, fill=black, draw=black, inner sep=1.5pt, right=0.5cm of v11] (v12) {};
		\node [inner, circle, fill=black, draw=black, inner sep=1.5pt, right=0.3cm of v12] (v13) {};
		
		\node [inner, circle, fill=black, draw=black, inner sep=1.5pt, right=0.5cm of v13] (v14) {};
		\node [inner, circle, fill=black, draw=black, inner sep=1.5pt, above=0.3cm of v13] (v0) {};
    \end{tikzpicture}};
	
    \foreach \from/\to in {v1/v2,v1/v3,v2/v4,v2/v3,v3/v4,v5/v6,v6/v7,v7/v8,v8/v5,v9/v10,v10/v11,v11/v9,v12/v13,v4/v5,v9/v12,v13/v14,v0/v14,v0/v12}
    \draw (\from) to (\to);
	
	\node at (2,6) [inner, circle, fill=white, draw=black, inner sep=1.5pt] (v15)  {};
	\node at (6,6) [inner, circle, fill=white, draw=black, inner sep=1.5pt] (v16)  {};
	
	\node at (1,5) [inner, circle, fill=white, draw=black, inner sep=1.5pt] (v17)  {};
	\node at (2,5) [inner, circle, fill=white, draw=black, inner sep=1.5pt] (v18)  {};
    \node at (3,5) [inner, circle, fill=gray, draw=gray, inner sep=1.5pt] (v19)  {};
	\node at (4,5) [inner, circle, fill=white, draw=black, inner sep=1.5pt] (v20)  {};
	\node at (5,5) [inner, circle, fill=white, draw=black, inner sep=1.5pt] (v21)  {};
	\node at (6,5) [inner, circle, fill=gray, draw=gray, inner sep=1.5pt] (v22)  {};
	\node at (7,5) [inner, circle, fill=white, draw=black, inner sep=1.5pt] (v23)  {};
	\node at (8,5) [inner, circle, fill=white, draw=black, inner sep=1.5pt] (v24)  {};
	\node at (9,5) [inner, circle, fill=white, draw=black, inner sep=1.5pt] (v25)  {};
	\node at (10,5) [inner, circle, fill=white, draw=black, inner sep=1.5pt] (v26)  {};
	\node at (11,5) [inner, circle, fill=white, draw=black, inner sep=1.5pt] (v27)  {};
	\node at (12,5) [inner, circle, fill=gray, draw=gray, inner sep=1.5pt] (v28)  {};
	\node at (13,5) [inner, circle, fill=white, draw=black, inner sep=1.5pt] (v29)  {};
	\node at (14,5) [inner, circle, fill=white, draw=black, inner sep=1.5pt] (v30)  {};
	\node at (15,5) [inner, circle, fill=gray, draw=gray, inner sep=1.5pt] (v31)  {};
	
	\node at (1,4) [inner, circle, fill=white, draw=black, inner sep=1.5pt] (v32)  {};
	\node at (7,4) [inner, circle, fill=white, draw=black, inner sep=1.5pt] (v33)  {};
	\node at (8,4) [inner, circle, fill=white, draw=black, inner sep=1.5pt] (v34)  {};
	\node at (9,4) [inner, circle, fill=white, draw=black, inner sep=1.5pt] (v35)  {};
	\node at (10,4) [inner, circle, fill=white, draw=black, inner sep=1.5pt] (v36)  {};
	\node at (11,4) [inner, circle, fill=white, draw=black, inner sep=1.5pt] (v37)  {};
	
	\node at (1,3) [inner, circle, fill=white, draw=black, inner sep=1.5pt] (v38)  {};
	\node at (2,3) [inner, circle, fill=white, draw=white, inner sep=1.5pt] (v39)  {};
	\node at (3,3) [inner, circle, fill=gray, draw=gray, inner sep=1.5pt] (v40)  {};
	\node at (4,3) [inner, circle, fill=white, draw=black, inner sep=1.5pt] (v41)  {};
	\node at (5,3) [inner, circle, fill=white, draw=black, inner sep=1.5pt] (v42)  {};
	\node at (6,3) [inner, circle, fill=gray, draw=gray, inner sep=1.5pt] (v43)  {};
	\node at (7,3) [inner, circle, fill=white, draw=white, inner sep=1.5pt] (v44)  {};
	\node at (8,3) [inner, circle, fill=white, draw=black, inner sep=1.5pt] (v45)  {};
	\node at (9,3) [inner, circle, fill=white, draw=black, inner sep=1.5pt] (v46)  {};
	\node at (10,3) [inner, circle, fill=white, draw=black, inner sep=1.5pt] (v47)  {};
	\node at (11,3) [inner, circle, fill=white, draw=black, inner sep=1.5pt] (v48)  {};
	\node at (12,3) [inner, circle, fill=gray, draw=gray, inner sep=1.5pt] (v49)  {};
	\node at (13.5,3) [inner, circle, fill=white, draw=black, inner sep=1.5pt] (v50)  {};
	\node at (15,3) [inner, circle, fill=gray, draw=gray, inner sep=1.5pt] (v51)  {};
	
	\node at (2,2) [inner, circle, fill=white, draw=black, inner sep=1.5pt] (v52)  {};
	\node at (3,2) [inner, circle, fill=white, draw=black, inner sep=1.5pt] (v53)  {};
	\node at (6,2) [inner, circle, fill=white, draw=black, inner sep=1.5pt] (v54)  {};
	\node at (8,2) [inner, circle, fill=white, draw=black, inner sep=1.5pt] (v55)  {};

	\node at (7,1.5) [inner, circle, fill=white, draw=black, inner sep=1.5pt] (v56)  {};
	\node at (9,1.5) [inner, circle, fill=white, draw=black, inner sep=1.5pt] (v57)  {};
	
    \foreach \from/\to in {v15/v19,v19/v18,v18/v17,v17/v32,v32/v38,v40/v52,v40/v53}
    \draw (\from) to (\to);
	
    \foreach \from/\to in {v19/v40,v19/v43,v40/v22,v22/v43,v28/v49,v28/v51,v49/v31,v31/v51}
    \draw (\from) to (\to);
	
    \foreach \from/\to in {v19/v20,v20/v21,v21/v22,v22/v23,v23/v24,v24/v25,v25/v26,v26/v27,v27/v28,v28/v29,v29/v30,v30/v31}
    \draw (\from) to (\to);
	
    \foreach \from/\to in {v43/v33,v33/v34,v34/v35,v35/v36,v36/v37,v37/v28}
    \draw (\from) to (\to);
	
    \foreach \from/\to in {v40/v41,v41/v42,v42/v43,v45/v46,v46/v47,v47/v48,v48/v49,v49/v50,v50/v51}
    \draw (\from) to (\to);
	
    \foreach \from/\to in {v16/v22,v54/v43,v45/v55,v55/v56,v55/v57}
    \draw (\from) to (\to);
	
	\foreach \from/\to in {v14/v30,v14/v31,v12/v49,v12/v28,v11/v25,v11/v34,v11/v47,v4/v40,v4/v42,v3/v15,v3/v19}
    \draw (\from) to (\to);
	
	\end{tikzpicture}
    \caption{A graph $G$ (not all edges shown), the set $S$ (in black), the set $R$ (in gray), and the set $T = G - R - S$ (in white).}
    \label{fig:1}
\end{figure}
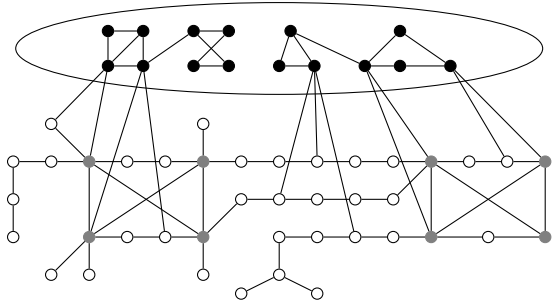

\begin{lemma}\label{lem-core}
Let $G = (V, E)$, $S$, $k$, and $g$ be as defined above. Let $R = \img(\reduce(G - S)) \subseteq (V \setminus S)$ 
denote the pre-image of $\reduce(G - S)$ in $G - S$. Then, $T = G - S - R$ is a forest and for 
every maximal degree-2 path in $\mathcal{P}$ there are at most two vertices on the path having 
neighbors in $R$ (in the graph $G - S$). Moreover, in time $k^{\OO(1)}\cdot|V|$, we can guarantee 
that $|T_{\leq 1}|$, $|\mathcal{P}|$, and $|T_{\geq 3}|$ are bounded by $k^{\OO(1)}$.
\end{lemma}

\begin{proof}
To see why $T = G - S - R$ must be a forest it is sufficient to note that for any cycle in $G - S$ at 
least one vertex from that cycle must be in $R = \img(\reduce(G - S))$ (see Figure~\ref{fig:1}). 
Recall that, since $\text{girth}(\reduce(G - S)) > 6$, every vertex in $R$ has degree at least 3 in $G - S$. 
Now assume there exists some path $P \in \mathcal{P}$ having exactly three (the same argument holds for any number) distinct 
vertices $u$, $v$ and $w$ (in that order) each having at least one neighbor in $R$ (possibly the same neighbor). 
We show that the middle vertex $v$ must have been in $R$, contradicting the fact that $T = G - S - R$. 
Consider the graph $G - S$ and apply Reduction Rules~\ref{rr-degone},~\ref{rr-degtwo} and~\ref{rr-multi} exhaustively (in $G - S$) on 
all vertices in the tree containing $P$ except for $u$, $v$ and $w$. Regardless of the order in which we apply 
the reduction rules, the path $P$ will eventually reduce to a path on three vertices, namely $u$, $v$, and $w$. 
To see why $v$ must be in $R$ observe that even if the other two vertices have degree two in the resulting graph, 
after reducing them, $v$ will have degree at least three (into $R$) and is therefore non-reducible.

Next, we bound the size of $T_{\leq 1}$, which implies a bound on the sizes of $T_{\geq 3}$ and $\mathcal{P}$. 
To do so, we simply invoke Lemma~\ref{lem-lossy} by setting $X = S \cup R$. Since $|S| < gk$, $g = {48 \log k \over \log \log k} + 6$ 
and $|R| \leq 3ck\log k + 2ck\log^{1.5} k$, we get that $|T_{\leq 1}| \leq |S \cup R|^2(2|S \cup R| + 1) = k^{\OO(1)}$. 
Since in a forest, it holds that $|T_{\geq 3}|<|T_{\leq 1}|$, the bound on $|T_{\geq 3}|$ follows. Moreover, in a forest, it also holds that 
$|\mathcal{P}|<|T_{\leq 1}|+|T_{\geq 3}|$ --- if we arbitrarily root each tree in the forest at a leaf, one end vertex of a path in $\mathcal{P}$ will be a parent 
of a different vertex from $T_{\leq 1}\cup T_{\geq 3}$ --- the bound on $|\mathcal{P}|$ follows as well.
\end{proof}


\section{Guessing Permutations}\label{sec:permutation}

This section is devoted to proving the following lemma. Note that assuming the 
statement of the lemma, the only remaining task (to prove Theorem \ref{thm:main}) is to develop an algorithm running 
in time $\OO(2^{|V|}\cdot\mathrm{poly}(|V|))$ and using polynomial space, which we present in Section~\ref{sec:DP}. 

\begin{lemma}\label{thm:permute}
Given an instance $(G, k)$ of {\sc Cycle Packing}, we can in time $2^{\OO(\frac{k\log^2k}{\log\log k})}\cdot|V|$ and polynomial 
space compute $2^{\OO(\frac{k\log^2k}{\log\log k})}$ instances of {\sc Cycle Packing} of the form $(G', k)$, where the 
number of vertices in $G'$ is bounded by $\OO(k\log^{1.5}k)$, such that $(G, k)$ is a yes-instance if 
and only if at least one of the instances $(G', k)$ is a yes-instance.\footnote{In practice, to use polynomial space, we output the instances one-by-one.}
\end{lemma}

\begin{proof}
We fix $g = {48 \log k \over \log \log k}+6$. Using Lemma~\ref{lem-packing}, we first compute 
a set $S$ in time $k^{\OO(1)}\cdot|V|$. Then, we guess which vertices to delete 
from $S$ --- that is, which vertices do not participate in a solution --- in time 
$\OO(2^{gk}) = 2^{\OO(\frac{k\log k}{\log\log k})}$. Here, guesses refer to different choices 
which lead to the construction of different instances of {\sc Cycle Packing} that 
are returned at the end (recall that we are allowed to return up to $2^{\OO(\frac{k\log^2k}{\log\log k})}$ different instances). 
Combining Lemma~\ref{lem-packing} and Corollary~\ref{cor-packing-size}, we now have a set 
$S \subseteq V$ such that $|S| = \OO({k \log k \over \log \log k})$, and $|\reduce(G - S)| = \OO(k\log^{1.5}k)$.

Applying Lemma~\ref{lem-core} with $R = \img(\reduce(G - S)) \subseteq (V \setminus S)$, we get a forest $T = G - (S\cup R)$ such that 
for every maximal degree-two path in $\mathcal{P}$ there are at most two vertices on the path having neighbors in $R$ (in the graph $G - S$). 
In addition, the size of $R$ is bounded by $\OO(k\log^{1.5}k)$, and the sizes $|T_{\leq 1}|$, $|\mathcal{P}|$ and $|T_{\geq 3}|$ are 
bounded by $k^{\OO(1)}$ (see Figure~\ref{fig:1}).

For every vertex in $S$ (which is assumed to participate in a solution), we now guess its two neighbors in a solution (see Figure~\ref{fig:example}). 
Note however that we only have a (polynomial in $k$) bound for $|S|$, $|R|$, $|T_{\leq 1}|$, $|\mathcal{P}|$ and $|T_{\geq 3}|$, but not for the length of 
paths in $\mathcal{P}$ and therefore not for the entire graph $G$. We let $Z_{\mathcal{P}}$ denote the set of vertices in $V(\mathcal{P})$ having neighbors in $R$. 
The size of $Z_{\mathcal{P}}$ is at most $2|\mathcal{P}|$. Moreover, we let $\mathcal{P}^\star$ denote the set of paths obtained after 
deleting $Z_{\mathcal{P}}$ from $\mathcal{P}$. Note that the size of $\mathcal{P}^\star$ is upper bounded by 
$|{\cal P}|+|Z_{\mathcal{P}}|\leq 3|{\cal P}|$, and that vertices in $V(\mathcal{P}^\star)$ are adjacent only 
to vertices in $V(\mathcal{P}^\star)\cup Z_{\mathcal{P}}\cup S$. Now, we create a set of ``objects'', $O = S \cup R \cup T_{\leq 1} \cup T_{\geq 3} \cup Z_{\mathcal{P}} \cup \mathcal{P}^\star$. 
We also denote $\widetilde{O}=O\setminus \mathcal{P}^\star$. We then guess, for each vertex in $S$, which two objects in $O$ 
constitute its neighbors, denoted by $\ell(v)$ and $r(v)$, in a solution. It is possible that $\ell(v)=r(v)$. 
Since $|O| = k^{\OO(1)}$, we can perform these guesses in $k^{\OO(\frac{k\log k}{\log\log k})}$, or equivalently $2^{\OO(\frac{k\log^2k}{\log\log k})}$, time. 
We can assume that if $\ell(v)\in \widetilde{O}$, then $\ell(v)$ is a neighbor of $v$, and otherwise $v$ has a neighbor on the path $\ell(v)$, else the current 
guess is not correct, and we need not try finding a solution subject to it. 
The same claim holds for $r(v)$. If $\ell(v)=r(v)\in \widetilde{O}$, then $\{v,\ell(v)\}$ is an edge of multiplicity two, and otherwise if $\ell(v)=r(v)$, then $v$ has (at least) two 
neighbors on the path $\ell(v)$.

Next, we fix some arbitrary order on $\mathcal{P}^\star$, and for each path in $\mathcal{P}^\star$, we fix some arbitrary orientation. 
We let $S^\star$ denote the multiset containing two occurrences of every vertex $v\in S$, denoted by $v_\ell$ and $v_r$. We guess an 
order of the vertices in $S^\star$. The time spent for guessing such an ordering is bounded by $|S|!$, which in turn is bounded by $2^{\OO(\frac{k\log^2k}{\log\log k})}$. 
The ordering, assuming it is guessed correctly, satisfies the following conditions. 
For each path $P\in\mathcal{P}^\star$, we let $\ell(P)$ and $r(P)$ denote the sets of vertices $v\in S$ such 
that $\ell(v)\in V(P)$ and $r(v)\in V(P)$, respectively. Now, for any two vertices $u,v\in \ell(P)$, if $u_\ell<v_\ell$ according to the order that we guessed, then 
the neighbor $\ell(u)$ of $u$ appears before the neighbor $\ell(v)$ of $v$ on $P$. Similarly, for any 
two vertices $u,v\in r(P)$, if $u_r<v_r$, then $r(u)$ appears before $r(v)$ on $P$. 
Finally, for any two vertices $u\in\ell(P)$ and $v\in r(P)$, if $u_\ell<v_r$, then $\ell(u)$ appears before $r(v)$ on $P$, and otherwise $r(v)$ appears before $\ell(u)$ on $P$.
%

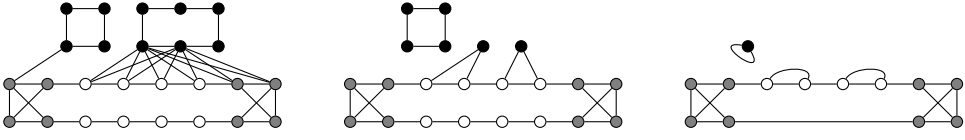
\begin{figure}[htp]
    \centering
    \begin{tikzpicture}[scale=.5, remember picture ,every node/.style={circle},inner/.style={circle},outer/.style={circle}]
    \node at (5.5,8) [inner, circle, fill=black, draw=black, inner sep=1.5pt] (v1)  {};
    \node at (6.5,8) [inner, circle, fill=black, draw=black, inner sep=1.5pt] (v2)  {};
    \node at (6.5,7) [inner, circle, fill=black, draw=black, inner sep=1.5pt] (v3)  {};
	\node at (5.5,7) [inner, circle, fill=black, draw=black, inner sep=1.5pt] (v4)  {};
    \node at (7.5,8) [inner, circle, fill=black, draw=black, inner sep=1.5pt] (v5)  {};
    \node at (8.5,8) [inner, circle, fill=black, draw=black, inner sep=1.5pt] (v6)  {};
    \node at (9.5,8) [inner, circle, fill=black, draw=black, inner sep=1.5pt] (v7)  {};
	\node at (7.5,7) [inner, circle, fill=black, draw=black, inner sep=1.5pt] (v8)  {};
    \node at (8.5,7) [inner, circle, fill=black, draw=black, inner sep=1.5pt] (v9)  {};
    \node at (9.5,7) [inner, circle, fill=black, draw=black, inner sep=1.5pt] (v10) {};
    \foreach \from/\to in {v1/v2,v2/v3,v3/v4,v4/v1,v5/v6,v6/v7,v7/v10,v10/v9,v9/v8,v8/v5}
    \draw (\from) to (\to);
	\node at (4,6) [inner, circle, fill=gray, draw=black, inner sep=1.5pt] (v11)  {};
	\node at (5,6) [inner, circle, fill=gray, draw=black, inner sep=1.5pt] (v12)  {};
	\node at (6,6) [inner, circle, fill=white, draw=black, inner sep=1.5pt] (v13)  {};
	\node at (7,6) [inner, circle, fill=white, draw=black, inner sep=1.5pt] (v14)  {};
	\node at (8,6) [inner, circle, fill=white, draw=black, inner sep=1.5pt] (v15)  {};
	\node at (9,6) [inner, circle, fill=white, draw=black, inner sep=1.5pt] (v16)  {};
	\node at (10,6) [inner, circle, fill=gray, draw=black, inner sep=1.5pt] (v17)  {};
	\node at (11,6) [inner, circle, fill=gray, draw=black, inner sep=1.5pt] (v18)  {};
	\node at (4,5) [inner, circle, fill=gray, draw=black, inner sep=1.5pt] (v19)  {};
	\node at (5,5) [inner, circle, fill=gray, draw=black, inner sep=1.5pt] (v20)  {};
	\node at (6,5) [inner, circle, fill=white, draw=black, inner sep=1.5pt] (v21)  {};
	\node at (7,5) [inner, circle, fill=white, draw=black, inner sep=1.5pt] (v22)  {};
	\node at (8,5) [inner, circle, fill=white, draw=black, inner sep=1.5pt] (v23)  {};
	\node at (9,5) [inner, circle, fill=white, draw=black, inner sep=1.5pt] (v24)  {};
	\node at (10,5) [inner, circle, fill=gray, draw=black, inner sep=1.5pt] (v25)  {};
	\node at (11,5) [inner, circle, fill=gray, draw=black, inner sep=1.5pt] (v26)  {};
    \foreach \from/\to in {v11/v12,v12/v13,v13/v14,v14/v15,v15/v16,v16/v17,v17/v18,v18/v26,v26/v25,v25/v24,v24/v23,v23/v22,v22/v21,v21/v20,v20/v19,v19/v11}
    \draw (\from) to (\to);
    \foreach \from/\to in {v11/v20,v12/v19,v17/v26,v25/v18}
    \draw (\from) to (\to);
    \foreach \from/\to in {v4/v11,v8/v13,v8/v14,v8/v15,v8/v16,v8/v17,v8/v18,v9/v13,v9/v14,v9/v15,v9/v16,v9/v17,v9/v18}
    \draw (\from) to (\to);
	\end{tikzpicture}
	\qquad
	\begin{tikzpicture}[scale=.5, remember picture ,every node/.style={circle},inner/.style={circle},outer/.style={circle}]
    \node at (5.5,8) [inner, circle, fill=black, draw=black, inner sep=1.5pt] (v1)  {};
    \node at (6.5,8) [inner, circle, fill=black, draw=black, inner sep=1.5pt] (v2)  {};
    \node at (6.5,7) [inner, circle, fill=black, draw=black, inner sep=1.5pt] (v3)  {};
	\node at (5.5,7) [inner, circle, fill=black, draw=black, inner sep=1.5pt] (v4)  {};
	\node at (7.5,7) [inner, circle, fill=black, draw=black, inner sep=1.5pt] (v8)  {};
    \node at (8.5,7) [inner, circle, fill=black, draw=black, inner sep=1.5pt] (v9)  {};
    \foreach \from/\to in {v1/v2,v2/v3,v3/v4,v4/v1}
    \draw (\from) to (\to);
	\node at (4,6) [inner, circle, fill=gray, draw=black, inner sep=1.5pt] (v11)  {};
	\node at (5,6) [inner, circle, fill=gray, draw=black, inner sep=1.5pt] (v12)  {};
	\node at (6,6) [inner, circle, fill=white, draw=black, inner sep=1.5pt] (v13)  {};
	\node at (7,6) [inner, circle, fill=white, draw=black, inner sep=1.5pt] (v14)  {};
	\node at (8,6) [inner, circle, fill=white, draw=black, inner sep=1.5pt] (v15)  {};
	\node at (9,6) [inner, circle, fill=white, draw=black, inner sep=1.5pt] (v16)  {};
	\node at (10,6) [inner, circle, fill=gray, draw=black, inner sep=1.5pt] (v17)  {};
	\node at (11,6) [inner, circle, fill=gray, draw=black, inner sep=1.5pt] (v18)  {};
	\node at (4,5) [inner, circle, fill=gray, draw=black, inner sep=1.5pt] (v19)  {};
	\node at (5,5) [inner, circle, fill=gray, draw=black, inner sep=1.5pt] (v20)  {};
	\node at (6,5) [inner, circle, fill=white, draw=black, inner sep=1.5pt] (v21)  {};
	\node at (7,5) [inner, circle, fill=white, draw=black, inner sep=1.5pt] (v22)  {};
	\node at (8,5) [inner, circle, fill=white, draw=black, inner sep=1.5pt] (v23)  {};
	\node at (9,5) [inner, circle, fill=white, draw=black, inner sep=1.5pt] (v24)  {};
	\node at (10,5) [inner, circle, fill=gray, draw=black, inner sep=1.5pt] (v25)  {};
	\node at (11,5) [inner, circle, fill=gray, draw=black, inner sep=1.5pt] (v26)  {};
    \foreach \from/\to in {v11/v12,v12/v13,v13/v14,v14/v15,v15/v16,v16/v17,v17/v18,v18/v26,v26/v25,v25/v24,v24/v23,v23/v22,v22/v21,v21/v20,v20/v19,v19/v11}
    \draw (\from) to (\to);
    \foreach \from/\to in {v11/v20,v12/v19,v17/v26,v25/v18}
    \draw (\from) to (\to);
    \foreach \from/\to in {v8/v13,v8/v14,v9/v15,v9/v16}
    \draw (\from) to (\to);
	\end{tikzpicture}
	\qquad
	\begin{tikzpicture}[scale=.5, remember picture ,every node/.style={circle},inner/.style={circle},outer/.style={circle}]
	\node at (5.5,7) [inner, circle, fill=black, draw=black, inner sep=1.5pt] (v1)  {};                                                  
	\draw (v1) to [out=530,in=300,looseness=8] (v1);
	\node at (4,6) [inner, circle, fill=gray, draw=black, inner sep=1.5pt] (v2)  {};
	\node at (5,6) [inner, circle, fill=gray, draw=black, inner sep=1.5pt] (v3)  {};
	\node at (6,6) [inner, circle, fill=white, draw=black, inner sep=1.5pt] (v4)  {};
	\node at (7,6) [inner, circle, fill=white, draw=black, inner sep=1.5pt] (v5)  {};
	\node at (8,6) [inner, circle, fill=white, draw=black, inner sep=1.5pt] (v6)  {};
	\node at (9,6) [inner, circle, fill=white, draw=black, inner sep=1.5pt] (v7)  {};
	\node at (10,6) [inner, circle, fill=gray, draw=black, inner sep=1.5pt] (v8)  {};
	\node at (11,6) [inner, circle, fill=gray, draw=black, inner sep=1.5pt] (v9)  {};
	\node at (4,5) [inner, circle, fill=gray, draw=black, inner sep=1.5pt] (v10)  {};
	\node at (5,5) [inner, circle, fill=gray, draw=black, inner sep=1.5pt] (v11)  {};
	\node at (10,5) [inner, circle, fill=gray, draw=black, inner sep=1.5pt] (v12)  {};
	\node at (11,5) [inner, circle, fill=gray, draw=black, inner sep=1.5pt] (v13)  {};
	\draw (v4) to [out=60,in=60] (v5);
	\draw (v6) to [out=60,in=60] (v7);
    \foreach \from/\to in {v2/v3,v3/v4,v4/v5,v5/v6,v6/v7,v7/v8,v8/v9,v9/v13,v13/v12,v12/v11,v11/v10,v10/v2,v2/v11,v10/v3,v13/v8,v9/v12}
	\draw (\from) to (\to);
	\end{tikzpicture}
    \caption{[Left] a graph $G$ (not all edges shown), the set $S$ (in black), the set $R$ (in gray), and the set $T = G - R - S$ (in white). [Center] the graph 
	obtained after guessing vertices in $S$ and their neighbors in a solution. [Right] example of a reduced instance.}
    \label{fig:example}
\end{figure}

Given a correct guess of $\ell(v)$ and $r(v)$, for each $v$ in $S$, as well as a correct guess of a 
permutation of $S^\star$, for each path in $\mathcal{P}^\star$, we let $\{x_v, y_v\}$ denote the two guessed neighbors 
of a vertex $v$ in $S$. Note that if $\ell(v)$ ($r(v)$) is in $\widetilde{O}$ then $x_v = \ell(v)$ ($y_v = r(v)$).
Otherwise, we assign neighbors to a vertex by a greedy procedure which agrees with the guessed permutation on $S^\star$; that is, for every path $P\in\mathcal{P}^\star$, we iterate over $\ell(P)\cup r(P)$ according to the guessed order, and for each vertex in it, assign its first neighbor on $P$ that is after the last vertex that has already been assigned (if such a vertex does not exist, we determine that the current guess is incorrect and proceed to the next one).
We let $X = \{x_v \mid v \in S\}$ and $Y = \{y_v \mid v \in S\}$. We also let $E_S$ be the set of edges incident on a vertex in $S$, 
and we let $E' = \{\{x_v, y_v\} \mid v \in S\}$ denote the set of all pairs of guesses.  
Finally, to obtain an instance $(G',k)$, we delete the vertex set $W = S \setminus (X \cup Y)$ from $G$, 
we delete the edge set $E_S$ from $G$, we add instead the set of edges $E'$, and finally we apply the 
reduce operator, i.e. $G' = \reduce((G - W - E_S) + E')$. 

\begin{claim}\label{cl1}
Let $(G',k)$ be one of the instances generated by the above procedure. Then, the number of vertices in $G'$ is bounded by $\OO(k\log^{1.5}k)$. 
\end{claim}

\begin{proof}
Recall that by Corollary~\ref{cor-packing-size}, we know that $|\reduce(G - S)| = \OO(k\log^{1.5}k)$. 
Moreover, we have $|E'| = |S| = \OO(\frac{k\log k}{\log\log k})$. 
Combining Observation~\ref{obs-edges} with the fact that $G' = \reduce((G - W - E_S) + E')$, we get 
$|\reduce((G - W - E_S) + E')| \leq |\reduce(G - W - E_S)| + 2|E'|$. 
Since in $G - W - E_S$ all vertices of $S$ have degree zero, $|\reduce(G - W - E_S)| \leq |\reduce(G - S)|$. 
Hence, we conclude that $|\reduce((G - W - E_S) + E')| = \OO(k\log^{1.5}k)$, as needed. 
\end{proof}

\begin{claim}\label{cl2}
$(G,k)$ is a yes-instance if and only if at least one of the generated instances $(G',k)$ is a yes-instance. 
\end{claim}

\begin{proof}
Assume that $(G,k)$ is a yes-instance and let $\mathcal{C} = \{C_1, C_2, \ldots\}$ be an optimal cycle packing, i.e set 
of maximum size of vertex-disjoint cycles, in $G$. Note that if no cycle in $\mathcal{C}$ intersects with $S$ then 
$\mathcal{C}$ is also an optimal cycle packing in $G - S$. By the safeness of our reduction rules, $\mathcal{C}$ is 
also an optimal cycle packing in $\reduce(G - S)$. Since we generate one instance for every possible intersection 
between an optimal solution and $S$, the case where no vertex from $S$ is picked corresponds to 
the instance $(G',k)$, with $G' = \reduce(G - S)$. Hence, in what follows we assume that some cycles in $\mathcal{C}$ 
intersect with $S$. Consider any cycle $C$ which intersects with $S$ and let $P_{C} = \{u_0, u_1, \ldots, u_f\}$ denote 
any path on this cycle such that $u_0,u_f \not\in S$ but $u_i \in S$ for $0 < i < f$. We claim that, for some $G'$, all such paths will be replaced 
by edges of the form $\{u_0, u_f\}$ in $\reduce((G - W - E_S) + E')$. Again, due to our exhaustive guessing, for some $G'$ we would have guessed, 
for each $i$, $\ell(u_i) = u_{i - 1}$ and $r(u_i) = u_{i + 1}$. Consequently, $P_{C} \setminus \{u_0, u_f\}$ is a degree-two path in 
$(G - W - E_S) + E'$ and therefore an edge in $\reduce((G - W - E_S) + E')$. Using similar arguments, 
it is easy to show that if $C$ is completely contained in $S$ then this cycle is contained in $G'$ as a loop on some vertex of the cycle. 

For the other direction, let $(G',k)$ be a yes-instance and let $\mathcal{C'} = \{C'_1, C'_2, \ldots\}$ be an optimal 
cycle packing in $G'$. We assume, without loss of generality, that $\mathcal{C'}$ is a cycle packing in $(G - W - E_S) + E'$, 
as one can trace back all reduction rules to obtain the graph $(G - W - E_S) + E'$. If no cycle in $\mathcal{C'}$ uses 
an edge $\{u_0, u_f\} \in E'$ then we are done, as $(G - W - E_S)$ is a subgraph of $G$. Otherwise, 
we claim that all such edges either exist in $G$ or can be replaced by vertex disjoint paths $P = \{u_0, u_1, \ldots, u_f\}$ 
(on at least three vertices) in $G$ such that $u_i \in S$ for $0 < i < f$. If either $u_0$ or $u_f$ is in $X \cup Y \subseteq S$ 
then the former case holds. It remains to prove the latter case. 
Recall that for every vertex in $S$ we guess its two neighbors from 
$O = S \cup R \cup T_{\leq 1} \cup T_{\geq 3} \cup Z_{\mathcal{P}} \cup \mathcal{P}^\star$. 
Hence, if $\{u_0, u_f\} \subseteq \widetilde{O}=O\setminus \mathcal{P}^\star$ then one can easily find 
a path (or singleton) in $G[S]$ to replace this edge by simply backtracking the neighborhood guesses. 
Now assume that $\{u_0, u_f\} \not\subseteq \widetilde{O}$ and recall that 
no vertex in a path in $\mathcal{P}^\star$ can have neighbors in $R$. 
Hence, any cycle containing such an edge must intersect with $S$ (in $G$). Assuming we have correctly guessed 
the neighbors of vertices in $S$ (as well as a permutation for $\mathcal{P}^\star$), we can again replace 
this edge with a path in $S$. 
%
\end{proof}

\noindent Combining Claims~\ref{cl1} and~\ref{cl2} concludes the proof of the theorem.
\end{proof}


\section{Dynamic Programming and Inclusion-Exclusion}\label{sec:DP}

Finally, we give an exact exponential-time algorithm for {\sc Cycle Packing}. For this purpose, we use 
DP and the principle of inclusion-exclusion, inspired by the work of Nederlof \cite{nederlof}. Due to space constraints, the details 
are given in Appendix \ref{app:DP}. 

\begin{lemma}\label{lem:DP}
There exists a (deterministic) polynomial-space algorithm that in time $\OO(2^{|V|}\cdot\mathrm{poly}(|V|))$ solves {\sc Cycle Packing}. 
In case a solution exists, it also outputs a solution.
\end{lemma}


\section{Conclusion}\label{sec:conclusion}

In this paper we have beaten the best known $2^{\OO(k\log^2k)}\cdot|V|$-time algorithm for {\sc Cycle Packing} that is a consequence of the Erdős-Pósa theorem. For this purpose, we developed a deterministic algorithm that solves {\sc Cycle Packing} in time $2^{\OO(\frac{k\log^2k}{\log\log k})}\cdot|V|$. Two additional advantageous properties of our algorithm is that its space complexity is polynomial in the input size and that in case a solution exists, it outputs a solution (in time $2^{\OO(\frac{k\log^2k}{\log\log k})}\cdot|V|$). Our technique relies on combinatorial arguments that may be of independent interest. These arguments allow us to translate any input instance of {\sc Cycle Packing} into $2^{\OO(\frac{k\log^2k}{\log\log k})}$ instances of {\sc Cycle Packing} whose sizes are small and can therefore be solved efficiently.

It remains an intriguing open question to discover the ``true'' running time,
under reasonable complexity-theoretic assumptions,
in which one can solve {\sc Cycle Packing} on general graphs. In particular, we would like to pose the following question: Does there exist a $2^{\OO(k\log k)}\cdot|V|^{\OO(1)}$-time algorithm for {\sc Cycle Packing}? This is true for graphs of bounded maximum degree as one can easily bound the number of vertices by $\OO(k\log k)$ and then apply Lemma~\ref{lem:DP}. Moreover, Bodlaender et al.~\cite{CycPacNoKernel} proved that this is also true in case one seeks $k$ edge-disjoint cycles rather than $k$ vertex-disjoint cycles. On the negative side, recall that (for general graphs) the bound $f(k)=\OO(k\log k)$ in the Erdős-Pósa theorem is essentially tight, and that it is unlikely that {\sc Cycle Packing} is solvable in time $2^{o(\mathrm{\it tw}\log \mathrm{\it tw})}\cdot|V|^{\OO(1)}$ \cite{Cut&Count}.\footnote{However, we do not rule out the existence of an algorithm solving {\sc Cycle Packing} in time $2^{\OO(\mathrm{\it fvs})}\cdot|V|^{\OO(1)}$.} Thus, the two most natural attempts to obtain a $2^{\OO(k\log k)}\cdot|V|^{\OO(1)}$-time algorithm -- either replacing the bound $\OO(k\log k)$ in the Erdős-Pósa theorem by $\OO(k)$ or speeding-up the computation based on DP to run in time $2^{\OO(\mathrm{\it tw})}\cdot|V|^{\OO(1)}$ -- lead to a dead end.

\bibliographystyle{splncs03}
\bibliography{References}

\newpage
\appendix

%
%

\section{Discarding Edges}\label{app:discard}

In this appendix we describe the procedure, mentioned in Section \ref{sec:prelim} to justify Assumption \ref{assume:numEdges}, which 
given a graph $G=(V,E)$, discards at least $|E|-(2c'k\log k + 1) \cdot |V|$ edges 
in time $\OO(k\log k\cdot|V|)$. We examine the vertices in $V$ in some arbitrary 
order $\{v_1,v_2,\ldots,v_{|V|}\}$, and initialize a counter $x$ to 0. For each 
vertex $v_i$, if $x<(2c'k\log k+1)\cdot |V|$ then we iterate over the set of edges incident 
to $v_i$, and for each edge whose other endpoint is $v_j$ for $j\geq i$, we increase $x$ by $1$. 
Let $\ell$ be the largest index for which we iterated over the set of edges incident to $v_\ell$. 
We copy $V$, and initialize the adjacency lists to be empty. Then, we copy the adjacency lists of the 
vertices $v_1,v_2,\ldots,v_\ell$, where for each adjacency involving vertices $v_i$ and $v_j$, where $i\leq\ell < j$, we update 
the adjacency list of $v_j$ to include $v_i$.

\section{Proof of Theorem \ref{thm-erdosposa}}\label{app:ErdosPosa}

We fix $c$ as the smallest integer such that $c\geq 150(\log_2c)$. Let $G=(V,E)$ be a (multi) graph, and let $k$ be a non-negative integer. 
The objective is to show that in time $k^{\OO(1)}\cdot|V|$ we can either output $k$ vertex-disjoint cycles 
or a feedback vertex set of size at most $ck \log k=r$. We remark that the first part of this proof, which ends at the statement 
of Lemma \ref{lem:cubicDiestel}, follows the proof of the Erdős-Pósa theorem \cite{ErdosPosa} given in the book~\cite{diestel-book}.

We may assume that $G$ contains at least one cycle, since this fact can clearly be checked in time $\OO(|V|+|E|)$, and if it is not true, we output 
an empty set as a feedback vertex set. Now, we construct a maximal subgraph $H$ of $G$ such each vertex in $H$ is of degree 2 or 3 (in $H$). 
This construction can be done in time $\OO(|V|+|E|)$ (see \cite{Max23Graph}). Let $V_2$ and $V_3$ be the degree-2 and degree-3 vertices in $H$, respectively.
We also compute (in time $\OO(|V|+|E|)$) the set ${\cal S}$ of connected components of $G-V(H)$. Observe that for each 
connected component $S\in{\cal S}$, there is at most one vertex $v_S\in V_2$ such that there is at least one vertex in $S$ adjacent to $v_S$, else we 
obtain a contradiction to the maximality of $H$ as it could have been extended by adding a path from $S$. We compute (in time $\OO(|V|+|E|)$) 
the vertices $v_S$, where for each component for which $v_S$ is undefined (since it does not exist), we set $v_S=nil$. Let $V^\star_2\subseteq V_2$ be the set 
of vertices $v_S\neq nil$ such that $v_S$ has at least two neighbors in $S$, which is easily found in time $\OO(|V|+|E|)$. Observe that 
if $|V^\star_2|\geq k$, we can output $k$ vertex-disjoint cycles in time $\OO(|V|+|E|)$. Thus, we next assume that $|V^\star_2|<k$. Moreover, observe 
that $V^\star_2\cup V_3$ is a feedback vertex set. Thus, if $|V^\star_2\cup V_3|\leq ck\log k$, we are done. We next assume 
that $|V^\star_2\cup V_3|>ck\log k$. In particular, it holds that $|V_3|>ck\log k - k\geq (c-1)k\log k$.

Let $H^*$ be the graph obtained from $H$ by contracting, for each vertex in $V_2$, an edge incident to it. 
We remark that here we permit the multiplicity of edges to be 3. Then, $H^*$ is a cubic graph whose vertex-set is $V_3$. To 
find $k$ vertex-disjoint cycles in $G$ in time $k^{\OO(1)}\cdot|V|$, it is sufficient to find $k$ vertex-disjoint cycles 
in $H^*$ in time $k^{\OO(1)}\cdot|V|$, since the cycles in $H^*$ can be translated into cycles 
in $G$ in time $\OO(|V|+|E|)$. We need to rely on the following claim, whose proof is given 
in the book~\cite{diestel-book}. We remark that the original claim refers to graphs, but it also holds for multigraphs. 

\begin{proposition}[\cite{diestel-book}]\label{lem:cubicDiestel}
If a cubic (multi) graph contains at least $q=4k(\log k + \log\log k + 4)$ vertices, then it contains $k$ vertex-disjoint cycles.
\end{proposition}

Thus, we know that $H^*$ contains $k$ vertex-disjoint cycles, and it remains to find them in time $k^{\OO(1)}\cdot|V|$. We now modify $H^*$ to obtain a cubic graph $H'$ on at least $q$ vertices but at most $\OO(k\cdot\log k)$ vertices, such that given $k$ vertex-disjoint cycles in $H'$, we can translate them into $k$ vertex-disjoint cycles in $H^*$ in time $\OO(|V|)$, which will complete the proof. To this end, we initially let $H'$ be a copy of $H^*$. Now, as long as $|V(H')|>(c-1)k\log k+2$, we perform the following procedure:
\begin{enumerate}
\item Choose arbitrarily a vertex $v\in V(H')$.
\item If $v$ has exactly one neighbor $u$ --- that is, $\{v,u\}$ is an edge of multiplicity 3 --- remove $v$ and $u$ from the graph.
\item Else if $v$ has a neighbor $u$ such that $u$, in turn, has a neighbor $w$ (which might be $v$) such that the edge $\{u,w\}$ is of multiplicity 2, then remove $u$ and $w$ from $H'$ and connect the remaining neighbor of $u$ to the remaining neighbor of $w$ by a new edge (which might be a self-loop).
\item Else, let $x,y,z$ be the three distinct neighbors of $v$. Then, remove $v$ and add an edge between $x$ and $y$. Now, each vertex is of degree 3, except for $z$, which is of degree 2, and has two distinct neighbors. Remove $z$, and connected its two neighbors by an edge.
\end{enumerate}

Since this procedure runs in time $\OO(1)$ and each call decreases the number of vertices in the graph, the entire process runs in time $\OO(|V|)$. It is also clear that the procedure outputs a cubic graph, and at its end, $(c-1)k\log k\leq |V(H')|\leq(c-1)k\log k+2$. Thus, to prove the correctness of the process, it is now sufficient to consider graphs $H_1$ and $H_2$, where $H_2$ is obtain from $H_1$ by applying the procedure once, and show that given a set ${\cal C}_2$ of $k$ vertex-disjoint cycles in $H_2$, we can modify them to obtain a set ${\cal C}_1$ of $k$ vertex-disjoint cycles in $H_1$. Let $v$ be the vertex chosen in the first step. If the condition in the second step was true, we simply let ${\cal C}_1={\cal C}_2$. In the second case, we examine whether the newly added edge belongs to a cycle in the solution in time $\OO(1)$ (as we assume that each element in the graph, if it belongs to the solution, has a pointer to its location in the solution), and if it is true, we replace it by the path between its endpoints whose only internal vertices are $u$ and $w$. Finally, suppose the procedure reached the last case. Then, if the first newly added edge is used, replace it by the path between its endpoints, $x$ and $y$, whose only internal vertex is $v$, and if the second newly added edge is used, replace it by the path between its endpoints whose only internal vertex is $z$.

We are now left with the task of finding $k$ vertex-disjoint cycles in $H'$. We initialize a set $\cal C$ of vertex-disjoint cycles to be empty. As long as $|{\cal C}| < k$, we find a shortest cycle in $H'$ in time $\OO(|V(H')|\cdot|E(H')|)=k^{\OO(1)}$ (see \cite{Itai}), insert it into $\cal C$ and remove all of the edges incident to its vertices from $H'$. Thus, to conclude the proof, it remains to show that for each $i\in\{0,1,\ldots,k-1\}$, after we remove the edges incident to the $i$th cycle from $H'$, it still contains a cycle.

By using induction on $i$, we show that after removing the edges incident to the $i$th cycle from $H'$, the number of edges in $H'$ is at least $p(i)=\frac{3}{2}(c-1)k\log_2 k - 12\cdot i\cdot \log_2(ck\log_2k)$. This would imply that the average degree of a vertex of $H'$ is at least $\displaystyle{\frac{2p(i)}{|V(H')|}\geq \frac{2p(i)}{(c-1)k\log_2 k+2}}\geq 2$ (we later also explicitly show that $2p(i)\geq (\sqrt{2}+1)ck\log_2 k$), and therefore it contains a cycle (since the average degree of a forest is smaller than 2). Initially, $H'$ is a cubic graph, and therefore $|E(H')|=\frac{3}{2}|V(H')|\geq \frac{3}{2}(c-1)k\log_2 k$, and the claim is true. Now, suppose that it is true for some $i\in\{0,1,\ldots, k-2\}$, and let us prove that it is true for $i+1$. By Proposition \ref{thm-alonetal}, a shortest cycle in $H'$ is of length at most $2 \log_{d-1} |V(H')| + 2\leq 3\log_{d-1}(ck\log_2 k)$, where $d=\displaystyle{\frac{2p(i)}{(c-1)k\log_2 k+2}} \geq \displaystyle{\frac{2p(i)}{ck\log_2 k}}$. Such a cycle is incident to at most $6\log_{d-1}(ck\log_2 k)$ edges. Therefore, after removing from $H'$ the edges incident to a shortest cycle in it, it contains at least $p(i)-6\log_{d-1}(ck\log_2 k) \geq p(i) - \displaystyle{6\frac{\log_2(ck\log_2 k)}{\log_2(d-1)}} = p(i) - \displaystyle{6\frac{\log_2(ck\log_2 k)}{\log_2(\frac{2p(i)}{ck\log_2 k}-1)}}$ edges. Thus, by the induction hypothesis, it remains to prove that $\displaystyle{\log_2(\frac{2p(i)}{ck\log_2 k}-1)} \geq 1/2$, to which end we need to show that $\displaystyle{\frac{2p(i)}{ck\log_2 k}-1}\geq \sqrt{2}$, that is, $2p(i)\geq (\sqrt{2}+1)ck\log_2 k$. For this purpose, it is sufficient to show that $4p(i)\geq 5ck\log_2 k$. By the induction hypothesis and since $i\leq k-1$, $4p(i)\geq 6(c-1)k\log_2 k - 48k\log_2(ck\log_2k) = 5ck\log_2k + (ck\log_2 k - 6k\log_2k -48k\log_2 k - 48k\log_2 c - 48k\log_2\log_2 k)\geq 5ck\log_2 k + (ck\log_2k - 150(\log_2 c)k\log_2k)$. Thus, we need to show that $c\geq 150(\log_2c)$, which holds by our choice of $c$. This concludes the proof.\qed

\section{Proof of Lemma \ref{thm-itai}}\label{app:thm-itai}

We can clearly detect self-loops and edges of multiplicity 2 in time $\OO(|V|+|E|)$, and return a cycle of length 1 or 2 accordingly, and therefore we next assume that $G$ is a simple graph.
Since $F$ is a feedback vertex set, to prove the lemma it is sufficient to present a procedure that given a vertex $v\in F$, finds in time $\OO(|V|+|E|)$ a cycle that is at least as short as the shortest cycle in $G$ that contains $v$. Indeed, then we can iterate over $F$ and invoke this procedure, returning the shortest cycle among those returned by the procedure. Thus, we next fix some vertex $v\in F$. Let $H$ be the connected component of $G$ containing $v$.

From the vertex $v$, we run a breadth first search (BFS). Thus, we obtain a BFS tree $T$ rooted at $v$, and each vertex in $V$ gets a level $i$, indicating the distance between this vertex and $v$ (the level of $v$ is 0). By iterating over the neighborhood of each vertex, we identify the smallest index $i_1$ such that there exists an edge with both endpoints, $u_1$ and $v_1$, at level $i_1$ (if such an index exists), and the smallest index $i_2$ such that there exists a vertex $w_2$ at level $i_2$ adjacent to two vertices, $u_2$ and $v_2$, at level $i_2-1$ (if such an index exists). For $i_1$, the edge $\{u_1,v_1\}$ and the paths between $v_1$ and $u_1$ and their lowest common ancestor result in a cycle of length at most $2i_1+1$. For $i_2$, the edges $\{w_2,u_2\}$ and $\{w_2,v_2\}$ and the paths between $u_2$ and $v_2$ and their lowest common ancestor result in a cycle of length at most $2i_2$. We return the shorter cycle among the two (if such a cycle exists).

Suppose that there exists a cycle containing $v$, and let $C$ be a shortest such cycle. We need to show that above procedure returns a cycle at least as short as $C$. Every edge of $H$ either connects two vertices of the same level, or a vertex of level $i-1$ with a vertex of level $i$. Thus, if there does not exist an index $i'_1$ such that there exists an edge in $E(C)$ with both endpoints, $u'_1$ and $v'_1$, at level $i'_1$, there must exist an index $i'_2$ such that there exists a vertex $w'_2$ at level $i'_2$ adjacent to two vertices, $u'_2$ and $v'_2$, at level $i'_2-1$, and the edges $\{w'_2,u'_2\}$ and $\{w'_2,v'_2\}$ belong to $E(C)$. First, suppose that the first case is true. Then, the procedure returns a cycle of length at most $2i'_1+1$. The length of $C$ cannot be shorter than $2i'_1+1$, since it consists of a path from $v$ to $u'_1$ (whose length is at least $i'_1$ since $u'_1$ belongs to level $i'_1$), a path from $v$ to $v'_1$ whose only common vertex with the previous path is $v$ (whose length is at least $i'_1$ since $v'_1$ belongs to level $i'_1$), and the edge $\{u'_1,v'_1\}$. Now, suppose that the second case is true. Then, the procedure returns a cycle of length at most $2i'_2$. The length of $C$ cannot be shorter than $2i'_2$, since it consists of two internally vertex-disjoint paths from $v$ to $w'_2$ (each of length at least $i'_2$ since $w'_2$ belongs to level $i'_2$).\qed

%
%

\section{Proof of Lemma \ref{lem-lossy}}\label{app:bound}
For $(u, v) \in X \times X$, let $L(u, v)$ be the set of vertices of degree at most one 
in $G - X$ such that each $x \in L(u, v)$ is adjacent to both $u$ and $v$ 
(if $u = v$, then $L(u, u)$ is the set of vertices which have degree at most one in $G - X$ and an edge of multiplicity two to $u$). 
For each pair $(u, v) \in X \times X$, we arbitrarily mark $2|X| + 1$ vertices 
from $L(u, v)$ if $|L(u, v)| > 2|X| + 1$, and we mark all vertices in $L(u, v)$ if $|L(u, v)| \leq 2|X| + 1$. We can execute 
this process as follows. First, in time $\OO(|X|\cdot|V|)$, for each vertex in $X$ we compute the set of its neighbors of degree at most one in $G-X$. 
Then, in time $\OO(|X|^3)$, for each pair $(u, v) \in X \times X$ we mark at most $2|X| + 1$ vertices as required.

Since we mark at most $2|X| + 1$ vertices for each pair $(u, v) \in X \times X$, there can 
be at most $|X|^2(2|X| + 1)$ marked vertices in $G - X$. Let $w$ be an unmarked vertex of degree at most one in $G - X$. 
We only consider the case where $\text{deg}_{G - X}(w) = 1$, as the other case can be proved analogously. 
Let $e$ be the unique edge in $G - X$ which is incident to $w$ and let $z$ be the other endpoint of this edge. 
Let $\mathcal{C}$ be a set of maximum size of vertex-disjoint cycles in $G$. 
Observe that if $\mathcal{C}$ does not contain a pair of cycles such that each of them intersects a 
different endpoint of $e$, then contracting $e$ keeps the resulting cycles vertex disjoint in $G/e$. 
Therefore, we may assume that $\mathcal{C}$ contains two cycles $C_w$ and $C_z$ where $C_w$ contains $w$ and $C_z$ contains $z$. 
The neighbor(s) of $w$ in $C_w$ must lie in $X$. Let these neighbors be $x$ and $y$ (again, $x$ and $y$ are not necessarily distinct). 
Since $w \in L(x, y)$ and it is unmarked, there are $2|X| + 1$ other vertices in $L(x, y)$ which were marked by the marking procedure. 
Moreover, each degree-1 vertex in $G - X$ that belongs to a cycle in ${\cal C}$ is either the predecessor 
or the successor of a vertex in $X$ in such a cycle. Therefore, at most $2|X|$ of the marked vertices 
can participate in cycles in ${\cal C}$. Hence, there exists a vertex in $L(x,y)$, call it $w'$, which is unused by ${\cal C}$.
Consequently, we can route the cycle $C_w$ through $w'$ instead of $w$, which gives 
us a set of $|\mathcal{C}|$ vertex disjoint cycles in $G/e$.

The first phase of the claimed $\OO(|X|^2\cdot k\log k\cdot |V|)$-time algorithm performs the above marking procedure, and then proceeds as follows. 
First it deletes every unmarked isolated vertex in $G - X$. Then, it contracts every edge in $G - X$ incident to at least one unmarked vertex of degree one in $G - X$. 
After these operations, new vertices in $G-X$ of degree at most one in $G-X$ might have been created. 
These vertices were either the unique neighbors in $G-X$ of deleted vertices or vertices incident to contracted edges. 
Thus, in case new vertices in $G-X$ of degree at most one in $G-X$ have been created, the algorithm performs another phase. 
Here, the algorithm iterates over the set of new vertices in $G-X$ of degree at most one in $G-X$, and for each such vertex, if it is a neighbor of two vertices in $X$ for which 
we have not yet marked $2|X| + 1$ vertices, the algorithm marks it. Then, the algorithm deletes vertices and contracts edges as it did in the first phase. 
The running time of such a phase is bounded by $\OO(|X|^2\cdot \rho)$, where $\rho$ is the total number of vertices deleted and edges contracted in the previous phase. 
As long as new degree-one vertices are created, the execution of the algorithm continues. Since each vertex can be deleted only once, and each edge can be 
contracted only once, the overall running time is bounded by $\OO(|X|(|X|^2+|V|) + |X|^2\cdot(|V|+|E|))=\OO(|X|^2\cdot k\log k\cdot |V|)$ 
(since $|E|=\OO(k\log k\cdot |V|)$). It also holds that when the 
algorithm terminates, $G-X$ contains at most $|X|^2(2|X| + 1)$ vertices whose degree in $G - X$ is at most one. 
This completes the proof of the lemma.\qed

%
%

\section{Proof of Lemma \ref{lem:DP}}\label{app:DP}

First, we recall the principle of inclusion-exclusion.

\begin{proposition}[Folklore, \cite{nederlof}]\label{lem:incExc}
Let $U$ and $R$ be sets, and for every $v\in R$ let $P_v$ be a subset of $U$. Use $\bar{P_v}$ to denote $U\setminus P_v$. With the convention $\bigcap_{v\in\emptyset}\bar{P_v} = U$, the following holds:
\[\displaystyle{|\bigcap_{v\in R}P_v| = \sum_{F\subseteq R}(-1)^{|F|}|\bigcap_{v\in F}\bar{P_v}|}.\]
\end{proposition}

We now proceed with the proof of Lemma \ref{lem:DP}. In the context of Proposition \ref{lem:incExc}, define the universe $U$ as the set of all tuples $(C_1,\ldots,C_k,$ $w^1_1,\ldots,w^1_k,w^2_1,\ldots,w^2_k,L)$ such that each $C_i$ is a closed walk in $G$ of length at least three, $w^1_i$ and $w^2_i$ are consecutive occurrences of vertices in $C_i$, $L\subseteq V$ and $(\sum_{i=1}^k|V(C_i)|)+|L|=|V|$. Here, by $|V(C_i)|$ we refer to a multiset -- that is, if $C_i$ contains $x$ occurrences of some vertex $v$, then $V(C_i)$ contains $x$ occurrences of $v$ as well. We define the requirement space $R=V$, and for each $v\in V$, we let $P_v$ be the set of all tuples $(C_1,\ldots,C_k,w^1_1,\ldots,w^1_k,w^2_1,\ldots,w^2_k,L)\in U$ such that $v\in(\bigcup_{i=1}^kV(C_i))\cup L$. On the one hand, if $G$ contains $k$ vertex-disjoint cycles $C_1,\ldots,C_k$, then for any choice of edges $\{w^1_1,w^2_1\}\in E(C_1),\ldots,\{w^1_k,w^2_k\}\in E(C_k)$, we have that $(C_1,\ldots,C_k,w^1_1,\ldots,w^1_k,w^2_1,\ldots,w^2_k,V\setminus(\bigcup_{i=1}^kV(C_i)))\in \bigcap_{v\in V}P_v$. On the other hand, if there exists $(C_1,\ldots,C_k,w^1_1,\ldots,w^1_k,w^2_1,\ldots,w^2_k,L)\in \bigcap_{v\in V}P_v$, then since $(\sum_{i=1}^k|V(C_i)|)+|L|=|V|$, each vertex $v\in V$ occurs exactly once in either exactly one of the closed walks $C_i$ or in the set $L$. In this case, we conclude that $C_1,\ldots,C_k$ are vertex-disjoint cycles. Therefore, we need to accept the input instance if and only if $|\bigcap_{v\in V}P_v|>0$.

By Proposition \ref{lem:incExc}, to decide whether $|\bigcap_{v\in V}P_v|>0$ in time $\OO(2^{|V|}\cdot\mathrm{poly}(|V|))$ and polynomial space, it is sufficient to show that for each subset $F\subseteq V$, $|\bigcap_{v\in F}\bar{P_v}|$ can be computed in polynomial time. To this end, we fix a subset $F\subseteq V$. Note that $\bigcap_{v\in F}\bar{P_v}$ is the set of all tuples $(C_1,\ldots,C_k,w^1_1,\ldots,w^1_k,w^2_1,\ldots,w^2_k,L)\in U$ such that $(\bigcup_{i=1}^kV(C_i))\cup L\subseteq V\setminus F$. Now, given an integer $\ell\in\{2k,\ldots,|V|\}$,\footnote{We do not consider the case where $\ell<2k$ since $k$ closed walks must overall contain at least $2k$ vertices.} let $Q_{\ell}$ denote the set of all tuples $(C_1,\ldots,C_k,w^1_1,\ldots,w^1_k,w^2_1,\ldots,w^2_k)$ such that each $C_i$ is a closed walk in $G-F$ of length at least three, $w^1_i$ and $w^2_i$ are consecutive occurrences of vertices in $C_i$, and $(\sum_{i=1}^k|V(C_i)|)=\ell$. Then, $|\bigcap_{v\in F}\bar{P_v}|=\sum_{\ell=2k}^{|V|}(|Q_\ell|\cdot\binom{|V\setminus F|}{|V|-\ell})$, where if $|V|-\ell<|V\setminus F|$, we let $\binom{|V\setminus F|}{|V|-\ell}=0$. Thus, it remains to show that each $|Q_\ell|$ can be computed in polynomial time. To this end, fix an integer $\ell\in\{2k,\ldots,|V|\}$.

Next, we will compute $|Q_\ell|$ by simply employing the method of dynamic programming. 
We use a matrix M that has an entry $[i,j,v,u]$ for all $i\in\{1,\ldots,k\}$, $j\in\{1,\ldots,\ell\}$ and $v,u\in V\setminus F$. 
Given $i\in\{1,\ldots,k\}$, $j\in\{1,\ldots,\ell\}$ and $v,u\in V\setminus F$, let $S(i,j,v,u)$ be the set of 
all tuples $(C_1,\ldots,C_i,w^1_1,\ldots,w^1_i,w^2_1,\ldots,w^2_i)$ such that for all $t\in\{1,\ldots,i-1\}$, $C_t$ is a closed walk of length 
at least three and $w^1_t$ and $w^2_t$ are consecutive occurrences of vertices in this walk, $C_i$ is a walk from $v$ to 
$u$, $w^1_i=v$ and $\sum_{t=1}^i|V(C_t)|= j$. The entry M$[i,j,v,u]$ will be used to store $|S(i,j,v,u)|$. Observe that
\[|Q_\ell|=\displaystyle{\sum_{v\in V\setminus F}\sum_{u\in N(v)\setminus F}\sum_{w\in N(u)\setminus F}|S(k,\ell-1,v,w)|}.\]
Thus, it remains to show that the entries of M can be calculated in polynomial time.

In the basis, we have the following calculations, relating to the case where $j=1$:
\begin{itemize}
\item If $j=1$ and ($i\geq 2$ or $v\neq u$): M$[i,j,v,u] = 0$.
\item Else if $j=1$: M$[i,j,v,u] = 1$.
\end{itemize}

Now, consider only entries where $j\geq 2$, which have not already been calculated in the basis. Then, we have the following calculations:
\begin{itemize}
\item If $i\geq 2, j\geq 3$ and $v=u$:\\M$[i,j,v,u] = \displaystyle{\sum_{w\in N(u)\setminus F}\mathrm{M}[i,j-1,v,w] + \sum_{p\in V\setminus F}\sum_{q\in N(p)\setminus F}\sum_{w\in N(q)\setminus F}\mathrm{M}[i-1,j-2,p,w]}$.
\item Else: M$[i,j,v,u] = \displaystyle{\sum_{w\in N(u)\setminus F}\mathrm{M}[i,j-1,v,w]}$.
\end{itemize}

It is straightforward to verify that the calculations are correct. The order of the computation is an 
ascending order with respect to $j$, which ensures that when an entry is calculated, the entries on which it relies have already been calculated.
To output a solution, we apply a simple self-reduction from the decision to the search variant of the problem. 
In particular, we repeatedly remove edges until no more edges can be removed from the graph while preserving a yes-instance.  
\qed


\end{document}